\documentclass[11pt]{article}

\usepackage[margin=1in]{geometry}

\usepackage{amsmath,amssymb,amsthm,ulem}
\usepackage{mathtools}
\usepackage{mathrsfs}
\usepackage{algorithm,algpseudocode}
\usepackage{graphicx}
\usepackage{xcolor}
\usepackage{array}
\usepackage{authblk,booktabs}
\usepackage{times,fourier,charter}
\usepackage[sort&compress]{natbib}
\newtheorem{theorem}{Theorem}[section]
\newtheorem{proposition}[theorem]{Proposition}

\theoremstyle{definition}

\newtheorem{remark}[theorem]{Remark}

\numberwithin{equation}{section}

\makeatletter
\renewcommand\section{\@startsection{section}{1}{\z@}%
  {-2.4ex \@plus -1ex \@minus -.2ex}%
  {1.2ex \@plus .1ex}%
  {\normalfont\bfseries\sffamily}}
\renewcommand\subsection{\@startsection{subsection}{2}{\z@}%
  {-1.8ex \@plus -0.4ex \@minus -.2ex}%
  {1.0ex \@plus .1ex}%
  {\normalfont\small\bfseries\sffamily}}
\renewcommand\subsubsection{\@startsection{subsubsection}{3}{\z@}%
  {-1.0ex \@plus -0.2ex \@minus -.2ex}%
  {0.8ex \@plus .1ex}%
  {\normalfont\normalsize\itshape}}
\renewcommand\paragraph{\@startsection{paragraph}{4}{\z@}%
  {0.6ex \@plus 0.2ex \@minus 0.1ex}%
  {-0.5em}%
  {\normalfont\normalsize\bfseries}}
\makeatother

\renewenvironment{thebibliography}[1]{%
  \begin{theoldbibliography}{#1}%
  \setlength{\itemsep}{2pt}%
  \setlength{\parskip}{2pt}%
}{%
  \end{theoldbibliography}%
}

\advance\textheight -2mm

\newcommand{\be}{\begin{equation}}
\newcommand{\ee}{\end{equation}}
\newcommand{\bse}{\begin{subequations}}
\newcommand{\ese}{\end{subequations}}

\newcommand{\dedication}{
  \begin{center}\itshape This work is dedicated to the memory of Francesco Calogero.
  \end{center}
}
\renewcommand{\boxed}[1]{#1}

\newcommand{\R}{\mathbb{R}}
\newcommand{\Q}{\mathbb{Q}}
\newcommand{\C}{\mathbb{C}}

\newcommand{\Z}{\mathbb{Z}}

\newcommand{\dx}{\partial_x}
\newcommand{\dt}{\partial_t}

\newcommand{\tr}{\operatorname{tr}}

\renewcommand{\thepage}{\arabic{page}}

\makeatletter
\renewcommand{\ps@plain}{%
  \renewcommand{\@oddhead}{}%
  \renewcommand{\@evenhead}{}%
  \renewcommand{\@oddfoot}{\hfil\small\thepage\hfil}%
  \renewcommand{\@evenfoot}{\hfil\small\thepage\hfil}%
}
\makeatother

\usepackage[colorlinks=true,allcolors=blue,pdfencoding=unicode]{hyperref}

\hypersetup{
  pdftitle={Learning Lax Pairs},
  pdfauthor={Jimmie Adriazola, Gino Biondini, Wei Zhu, Panayotis G. Kevrekidis},
  pdfsubject={Integrable systems, Lax pairs, sparse identification},
  pdfkeywords={Lax pairs, integrable systems, SILO, KdV, shallow water, Burgers}
}

\title{\texorpdfstring{Learning Lax Pairs: Revisiting the Classical Paradigm}{Learning Lax Pairs }}

\author[1,2]{Jimmie Adriazola}
\author[3]{Gino Biondini}
\author[4]{Wei Zhu}
\author[5,6,7]{Panayotis~G.~Kevrekidis}

\affil[1]{School of Mathematical and Statistical Sciences, Arizona State University, Tempe, AZ}
\affil[2]{Simon A. Levin Mathematical, Computational and Modeling Sciences Center, Arizona State University, AZ}
\affil[3]{Department of Mathematics, State University of New York at Buffalo, Buffalo, NY}
\affil[4]{School of Mathematics, Georgia Institute of Technology, Atlanta, GA}
\affil[5]{Department of Mathematics and Statistics, University of Massachusetts Amherst, 01003 Amherst, MA}
\affil[6]{Department of Physics, University of Massachusetts Amherst, 01003 Amherst, MA}
\affil[7]{Department of Mechanical Engineering, Seoul National University,
1 Gwanak-ro, Gwanak-gu, Seoul 08826, South Korea}

\date{\small\today}

\begin{document}
\maketitle
\dedication
\begin{abstract}
A Lax pair $(L,P)$ is sometimes thought of as a structural certificate, in that
the spatial operator $L$ carries the spectral data of an integrable system, and
its isospectral evolution under $\partial_t L = [L,P]$ encodes the nonlinear
dynamics. Yet, experience shows that the correspondence between equations and Lax
pairs is much more nuanced than this picture suggests. Equations can admit Lax
pairs that fail to encode the expected integrable structure. This paper probes
that anomalous corner of the Lax pair landscape through five case studies (the
Euler top, the free Schr\"odinger equation, the inviscid Burgers equation, the
shallow water system, and the Korteweg--de Vries equation), each illustrating a
different way the link to integrability can be distorted. The approach combines
analytical calculations with the Sparse Identification of Lax Operators (SILO)
framework, which proved useful throughout, in some cases confirming the textbook
pair and in others surfacing alternatives worth understanding on their own terms.
The recurring lesson across the five cases is that compatibility underdetermines
the Lax representation, so that anomalous pairs are regular features of the
landscape rather than pathologies. Notably, we show that a spectrally degenerate
Korteweg--de Vries Lax pair, classified as fake by standard criteria, still
generates the full conservation hierarchy through its operator algebra, which
shows that a blunt dichotomy between true and fake Lax pairs can be too reductive.
\end{abstract}
\par\smallskip

\section{Introduction}

The standard picture of a Lax pair comes in several forms~\cite{Lax1968,AKNS1974,FaddeevTakhtajan,BabelonBernardTalon2003,Doikou2012}. In the original operator version~\cite{Lax1968}, one seeks operators $L$ and $P$ acting on a fixed Hilbert space and depending on a field $u$, such that
\begin{equation}
\label{eq:Lax}
\dt L = [L, P]
\end{equation}
is equivalent to the evolution equation for $u$. The isospectral flow generated by \eqref{eq:Lax} produces an infinite family of conserved quantities through the spectral data of $L$, and, under favorable circumstances, opens the door to the inverse scattering transform \cite{Gardner1967, AblowitzSegur}. The same compatibility logic applies when $L$ and $P$ are finite-dimensional
matrices (as in the Manakov pair for the Euler top~\cite{Manakov1976}), symbols
on phase space (with the commutator replaced by a Poisson bracket, as in
dispersionless hierarchies~\cite{Zakharov1994}), or connection forms (with the
commutator replaced by a curvature, as in zero-curvature
representations~\cite{ZakharovShabat}). The correspondence between integrable
equations and Lax pairs, in any of these forms, has been a productive organizing
principle in integrable systems theory~\cite{FaddeevTakhtajan, AKNS1974,Dickey2003,DrazinJohnson1989}, though
one with subtleties~\cite{CalogeroNucci}.

What the standard picture leaves implicit is that the same equation can admit different Lax pairs, and the same Lax pair can be compatible with different equations. For example, a single integrable equation typically admits many inequivalent Lax pairs, related by gauge transformations~\cite{FaddeevTakhtajan}, spectral
shifts~\cite{Manakov1976, MishchenkoFomenko1978}, or less structured
redundancies~\cite{CalogeroNucci}. Some of these encode the same spectral content in different coordinates. Others do not, as was already recognized by Calogero and Nucci's ``Lax pairs galore'' \cite{CalogeroNucci}. Calogero and Nucci showed that a pair satisfying \eqref{eq:Lax} need not certify integrability at all, and Sakovich~\cite{Sakovich2001}, and later Butler and Hay \cite{ButlerHay}, gave systematic methods to detect such ``fake'' Lax pairs. Following these previous works, we call a pair \emph{fake} if it
satisfies a Lax compatibility relation, such as~\eqref{eq:Lax}, on solutions of the
equation but its spectral data certify no integrable structure and support no
inverse scattering transform. Additionally and throughout, we use the word \emph{anomalous}  as an umbrella term for any pair that
satisfies a Lax compatibility relation yet departs from the classical spectral
picture, whether by spectral collapse or by hiding an entire parametric family
behind a single base pair. We use it informally and indicate the individual departures as they
arise throughout the paper.

The phenomenon has practical consequences.  Anomalous  pairs have been identified in discrete integrable equations \cite{Gubbiotti2016}, in modified Korteweg--de Vries (KdV) families \cite{axioms13020121}, in the 2D and 3D Euler equations of incompressible fluids \cite{FriedlanderVishik1990, Li2001Euler2D, Childress3DEuler},
and in the broader Calogero-Nucci construction that produces a pair from any conservation law \cite{CalogeroNucci}. Any method that tries to discover a Lax pair~\cite{SILO,
KrippendorfLustSyvaeri2021, LinChen2025LaxPairFIND, deKosterWahls2024} as a
structural object for scientific applications, whether for the nonlinear Fourier
transform developed for fiber-optic communication~\cite{TuritsynPrilepskyEtAl2017,
Kotlyar2020NFTneural, Sedov2025NFTNumerics}, the nonlinear-spectrum analysis of
measured ocean records, which classifies rogue waves~\cite{TeutschBuhrenWaseda2023,
Osborne2020NLFA} and has lately resolved the first soliton-gas sea states observed
in the deep open ocean~\cite{LeeWahls2025}, reduced-order modeling along the lines
of~\cite{GerbeauLombardi2014}, neural surrogates that exploit
Lax structure~\cite{PuChen2024LPNN}, symbolic-computation pipelines that produce
Lax pairs algorithmically~\cite{BridgmanHeremanQuispelvanderKamp2013}, or 
in a principled data-driven fashion~\cite{SILO}
the practitioner has to reckon with the
possibility that the discovered pair is useless.

In parallel, the existence of Lax-like structures that do not conform to the classical spectral picture is, in places, already well developed in the literature. For KdV alone, Miura, Gardner, and Kruskal \cite{Miura1968,MiuraGardnerKruskal1968} showed that the conservation hierarchy follows from a generating-function argument that does not require the Schr\"odinger spectral problem. The Lambert--Musette Bell-polynomial formalism \cite{LambertSpringael2008,LambertLebleSpringael2001,Fan2011,GilsonLambertNimmoWillox1996} recasts bilinearizations, B\"acklund transformations, and Lax pairs of KdV-type equations through Bell polynomials, exposing algebraic content that the operator-theoretic presentation does not make visible. In the pseudopotential and prolongation-structure tradition \cite{WahlquistEstabrook1975}, first-order scalar auxiliary equations generate conservation laws for a range of integrable PDEs without passing through a second-order operator. 

Our work sits alongside these developments and 
revisits the landscape of Lax representations 
via a computational approach using the 
Sparse Identification of Lax Operators (SILO) framework, developed in \cite{SILO} and revisited here.
Concentrating our attention on five selected case studies, we investigate
the algebraic and structural features that distinguish one representation from another, and on what a computational search turns up when it is not constrained to land on a pair that fits a pre-existing template.
The SILO framework formulates the search for Lax pairs as a structured sparse regression problem: a library of candidate operators is combined against the compatibility residual of \eqref{eq:Lax}, in any of its forms, evaluated pointwise in phase or parameter space through the chain rule, with sparsity regularization selecting parsimonious solutions and problem-specific penalties preventing degenerate ones. At the same time, we emphasize that 
SILO is not the only tool of this paper; it is an instrument that paves numerous 
further directions.
For instance, in what follows, SILO may lead to the straightforward verification of a classical pair, it may spearhead the surfacing of genuinely anomalous Lax pair structure 
or it may suggest some Lax pair modification that is subsequently worked out by hand.

The paper is organized as follows.  Section~\ref{sec:euler} treats the Euler top, where an unconstrained sparse search recovers a dynamically correct but spectrally trivial pair whose invariants fail to generate the full Liouville structure, a finite-dimensional instance of the fake Lax pair phenomenon in its cleanest form. A Manakov-type spectral shift \cite{Manakov1976} rescues the situation, and the shift itself is computed by linear algebra rather than further numerics. Section~\ref{sec:schrodinger} studies the free Schr\"odinger equation, where two different nondegeneracy penalties in the SILO loss return two structurally different pairs, the textbook pair~\cite{Fokas2008} and another, spectrally degenerate pair in a precise algebraic sense. Section~\ref{sec:burgers} turns to the inviscid Burgers equation, where a symbolic regression on the Clairaut compatibility identity led us to an infinite-dimensional gauge freedom parametrized by an arbitrary function of one variable, and across spatial orders to a continuum $S_x^m = au + \lambda$ of auxiliary systems valid for every $m \in \mathbb{R}_{>0}$. Section~\ref{sec:shallowwater} treats the shallow water equations on two registers. Within the Laurent layer of dispersionless Lax functions, the classical Brunelli--Das representation~\cite{BrunelliDas1997} is the only compatible pair our sparse search returned. Outside the Laurent layer, the Calogero--Nucci hodograph mechanism applied to an Euler--Poisson--Darboux equation in Riemann coordinates produces an infinite-dimensional family of compatible auxiliary systems, encompassing the polynomial Whitham hierarchy of conserved densities, an algebraic spectral-discriminant continuum, and a dressing transform that links them.

Section~\ref{sec:kdv} focuses on the first-order scalar operator $L = u + \dx$ that has recently been revealed for the KdV equation via SILO in \cite{SILO}. Although, to the best of our knowledge, this pair does not appear in the literature, the existence of such a pair for any scalar equation in conservation form $(u_t=q_x)$ is, in fact, classical. This traces to separate discussions with Manakov and Ablowitz in the 1980s~\cite{KonopelchenkoPC}, transmitted through the Calogero--Nucci construction~\cite{CalogeroNucci}, where the prevailing verdict was that these Lax pairs were useless.  Here we show that, while the spectrum of $L$ is trivial by every classical measure (empty $L^2$ point spectrum, periodic eigenvalues encoding only the mean of $u$, monodromy reducing to the mass), powers of $L$ nevertheless generate a ring of differential polynomials, the complete exponential Bell polynomials in the derivatives of $u$, and the classical KdV conservation laws appear as the kernel of the KdV time derivative on this ring at successive subscript-sum levels. The Bell polynomial connection to KdV-type equations is well developed in the Lambert--Musette--Springael formalism \cite{LambertSpringael2008}, where these polynomials enter through logarithmic linearization in the variable $v = \ln \phi$ with $\phi$ a Schr\"odinger eigenfunction, so that $v_x$ plays the role of a Riccati pseudopotential. Our framing is different, however. Taking $v = \int u\,dx$, the gauge identity $L^s = e^{-v}\partial_x^s e^v$ identifies the Bell polynomials with the zeroth-order parts of $L^s$ acting on the constant function. Through this identification, the Lax equation $\partial_t L = [L,P]$ descends to a $\mathbb{Q}$-linear derivation on the ring of integrated $\rho$-monomials modulo total $x$-derivatives. Its kernel at each subscript-sum level is the KdV conservation hierarchy, extracted by finite-dimensional linear algebra over~$\mathbb{Q}$. 

We suggest keeping in mind the following while reading this work. First, this is not a paper about discovering integrability. Every equation considered here is classically known to be integrable (or, in the Burgers case, linearizable). The object of study is the Lax representation itself: what forms it can take, what it can and cannot encode, and how algebraic and spectral content trade off across different Lax gauges. Second, SILO appears throughout the paper but the method itself is not the central point. Readers interested in SILO as a methodology are directed to \cite{SILO}, where the framework is carefully developed from the ground up. The broader argument this paper makes is that anomalous Lax pairs are regular features of the integrable landscape that any data-driven discovery or machine learning pipeline will inevitably encounter, and whose structure deserves to be understood on its own terms if future searching is to land on pairs that carry real content for scientific applications.

\section{The Euler Top: Fake Lax Pairs in Finite Dimensions}
\label{sec:euler}

The cleanest finite-dimensional illustration of an anomalous Lax pair lives within the Euler top. The system is simple, the classical Lax pair is well known, and yet an unconstrained search for Lax representations returns a pair that satisfies the Lax equation exactly while carrying strictly less information than the classical one. The phenomenon is a genuine feature of the Lax correspondence in the finite-dimensional Lie--Poisson setting, and it admits a classical remedy through a Manakov-type spectral shift \cite{Manakov1976}. We present both the degenerate pair and the remedy here, and in doing so fix the vocabulary of fake, degenerate, and spectrally incomplete that will recur throughout the paper.

\subsection{The system and its classical Lax representation}

The angular velocities $\Omega_i$ of the Euler top satisfy
\begin{equation}\label{eq:EulerTop}
    I_i \dot{\Omega}_i = (I_j - I_k)\, \Omega_j \Omega_k, \qquad (i,j,k)\ \text{cyclic in}\ (1,2,3),
\end{equation}
where $I_i > 0$ are the principal moments of inertia. As a Hamiltonian system evolving on the dual Lie algebra $\mathfrak{so}(3)^*$, the Euler top admits conserved quantities tied directly to its underlying geometric structure. In addition to the kinetic energy, the flow preserves the quadratic Casimir $\|M\|^2 := \sum_{i=1}^3 M_i^2$. The preservation of $\|M\|^2$ reflects the degeneracy of the Lie--Poisson structure on $\mathfrak{so}(3)^*$, which foliates phase space into symplectic leaves given by the level sets of the Casimir. Together, these invariants constrain the motion to the intersection of two quadratic surfaces in $\mathbb{R}^3$, providing a geometric manifestation of integrability and sharply restricting the admissible dynamics~\cite{kozlov1993,bolsinovfomenko2004}.

The Euler top admits a Lax pair~\cite{Manakov1976}. Indeed, Equation~\eqref{eq:EulerTop} is compatible with
\begin{equation}\label{eq:KnownEulPair}
L=\begin{pmatrix}
0 & -M_3 & M_2\\
M_3 & 0 & -M_1\\
-M_2 & M_1 & 0
\end{pmatrix}, \qquad
P=\begin{pmatrix}
0 & -\Omega_3 & \Omega_2\\
\Omega_3 & 0 & -\Omega_1\\
-\Omega_2 & \Omega_1 & 0
\end{pmatrix},
\end{equation}
with $M_i = I_i \Omega_i$, so that the Lax equation $\dot L = [L, P]$ reproduces the dynamics. Although this formulation correctly encodes the equations of motion, it does not capture the full integrable structure of the Euler top.  The Euler top has two functionally independent first integrals: the
Hamiltonian $H = \tfrac12 \sum_{i=1}^3 I_i \Omega_i^2 = \tfrac12 \sum_{i=1}^3 M_i^2 / I_i$
and the Casimir $\|M\|^2 = \sum_{i=1}^3 M_i^2$. On $\mathfrak{so}(3)^* \cong \R^3$
the symplectic leaves are the spheres $\|M\|^2 = \text{const}$, each of dimension
two, so a single integral in involution suffices for Liouville integrability on a
leaf, and $H$ plays that role. 

The pair \eqref{eq:KnownEulPair} reproduces the
equations of motion, but its spectral data carries less information, that is, it recovers the Casimir $\|M\|^2$ but not the
Hamiltonian $H$.
Observe that the matrix $L$ is
skew-symmetric with eigenvalues $0$ and $\pm i\|M\|$, so $\operatorname{tr}(L^2) = -2\|M\|^2$
and every higher trace invariant is determined by $\|M\|^2$. The Hamiltonian,
which weighs the $M_i^2$ by $1/I_i$, is not among them. The pair \eqref{eq:KnownEulPair} is therefore spectrally incomplete. Its Lax
equation $\dot L = [L, P]$ reproduces the equations of motion, so the dynamics
is correct, but the invariants of $L$ recover only $\|M\|^2$, which fixes the
sphere $\|M\|^2 = \mathrm{constant}$ on which the motion lies, and never $H$, the
integral that selects the trajectory within it. Recovering the Casimir but not
the Hamiltonian does not certify Liouville integrability
\cite{reyman1980, bolsinovfomenko2004}.

The classical remedy, due to Manakov \cite{Manakov1976}, introduces a diagonal matrix
\begin{equation}\label{eq:ManakovJ}
J = \tfrac12\operatorname{diag}\bigl(I_2 + I_3 - I_1,\ I_1 + I_3 - I_2,\ I_1 + I_2 - I_3\bigr),
\end{equation}
under which the shifted Lax equation
\[
\frac{d}{dt}(L + \lambda J^2) = \bigl[L + \lambda J^2,\ P + \lambda J\bigr]
\]
holds identically in $\lambda \in \mathbb{C}$. The matrix $\hat L = L + \lambda J^2$
is then isospectral, so every $\operatorname{tr}\hat L^n$ is conserved, and the two
lowest traces already carry both integrals: $\operatorname{tr}\hat L^2$ returns the
Casimir $\|M\|^2$ and $\operatorname{tr}\hat L^3$ returns the Hamiltonian $H$ (the
computation is collected at the end of this subsection). Higher traces add nothing
functionally independent on the three-dimensional top. The spectral parameter is
what lifts the pair \eqref{eq:KnownEulPair} from a dynamical certificate to a
Liouville certificate. This is the first structural lesson we encounter: \emph{a Lax
pair and its spectral completion are different objects}. We also learn that nothing here demands that every pair admit a
spectral completion.

\subsection{Naive searches yield rank-deficient Lax pairs}

For what follows, it is useful to express $\dot L$ through its Hamiltonian chain rule. The Hamiltonian $H = \frac12\sum_{i=1}^3 I_i\Omega_i^2$ generates the equations of motion via the Lie--Poisson bracket, written here in the convention $\dot f=\{H,f\}$, so that $\dot\Omega_i=\{H,\Omega_i\}$, where $\{\cdot,\cdot\}$ denotes the bracket on $\mathfrak{so}(3)^*$ pulled back through $M_i=I_i\Omega_i$. Applying the chain rule and noting that the moments of inertia $I_i$ are constant,
\begin{equation}\label{eq:EulerChain}
\frac{dL}{dt}
= \sum_{i=1}^3 \frac{\partial L}{\partial \Omega_i}\,\dot\Omega_i
= \sum_{i=1}^3 \frac{\partial L}{\partial \Omega_i}\,\{H, \Omega_i\}
= [L, P].
\end{equation}
This valuable identity expresses the Lax equation as a pointwise constraint in phase space. It avoids time discretization entirely and converts the search for Lax pairs into an operator-identification problem evaluated on samples of $\Omega$. Every sparse regression in this paper uses some version of this identity, either directly or in spirit.

Now, suppose we assume nothing about the structure of the Lax pair beyond a linear dependence on the state variables,
\begin{equation}\label{eq:SimpEulHyp}
\tilde L_{i,j} = \sum_k \xi_{i,j,k} \Omega_k, \qquad \tilde P_{i,j} = \sum_k \zeta_{i,j,k} \Omega_k,
\end{equation}
and minimize the natural relative residual
\[
\mathcal{J}_{\rm loss}(\tilde L, \tilde P) = \mathbb{E}_{\Omega \sim \mathcal{P}}\,
\frac{\sum_{j,k} \bigl(\sum_i \frac{\partial \tilde L}{\partial \Omega_i} \{H, \Omega_i\} - [\tilde L, \tilde P]\bigr)_{j,k}^2}{\sum_{j,k} \bigl(\sum_i \frac{\partial \tilde L}{\partial \Omega_i} \{H, \Omega_i\}\bigr)_{j,k}^2}
\]
with a sparsity penalty,
\begin{equation}\label{eq:EulerNatural}
\min_{\eta \in \mathbb{R}^N} \mathcal{J}_{\rm natural}[\eta] := \min_{\eta \in \mathbb{R}^N} \mathcal{J}_{\rm loss}[\eta] + s\,\mathcal{S}[\eta],
\end{equation}
where $\mathcal{P}$ denotes the uniform distribution on the unit cube centered at the origin, $\eta$ collects the $54$ parameters $\xi_{i,j,k}, \zeta_{i,j,k}$, and $\mathcal{S}$ can be any viable sparsity functional, e.g., $\|\eta\|_{{l^1}(\mathbb{R}^{N})}$. The denominator normalization suppresses the trivial solution $\tilde L \equiv 0$. 

What comes back is not the classical pair \eqref{eq:KnownEulPair}. Instead, repeatedly and across random initializations, SILO returns rank deficient Lax pairs. For purposes of this discussion, one such pair is
\begin{equation}\label{eq:degenEulLax}
L_{\rm deficient} = \begin{pmatrix}
0 & 0 & 0 \\
0 & \Omega_1 & \Omega_2 \\
0 & \dfrac{I_2 (I_3 - I_2)}{I_1 (I_3 - I_1)}\,\Omega_2 & -\Omega_1
\end{pmatrix}, \quad
P_{\rm deficient} = \begin{pmatrix}
0 & 0 & 0 \\
0 & 0 & \dfrac{I_3 - I_1}{2 I_2}\,\Omega_3 \\
0 & \dfrac{I_2 - I_3}{2 I_1}\,\Omega_3 & 0
\end{pmatrix}.
\end{equation}
A direct calculation confirms that this pair satisfies $\dot L_{\rm deficient} = [L_{\rm deficient}, P_{\rm deficient}]$ exactly, with only $6$ of $54$ parameters active. From the perspective of $\mathcal{J}_{\rm natural}$, it is the best possible solution: perfect residual with maximal sparsity. From the perspective of integrability, it is useless. The matrix $L_{\rm deficient}$ is rank-two, $\Omega_3$ does not appear in it at all, and its spectral data do not recover either the Casimir or the Hamiltonian.

This is an even worse dilemma than being a \emph{fake Lax pair} in the strict sense of \cite{ButlerHay, Gubbiotti2016}: the compatibility condition is satisfied, yet the integrable structure is not certified \emph{and} the dynamics are not even reproduced. We believe that this phenomenon is really a consequence of a bad search geometry, i.e., a generic necessary-but-not-sufficient objective that is given a generic hypothesis. These generic searches consistently trip over rank-deficient pairs that populate the Lax landscape. Clearly, the route to discovery of more meaningful Lax pairs requires more principled searching. 

\subsection{Recovering the classical structure}

Two structural constraints suffice to rule out \eqref{eq:degenEulLax} and restore the classical pair. First, we enforce skew-symmetry in \eqref{eq:SimpEulHyp}, i.e., $\tilde{L}=-\tilde{L}^\top$ and $\tilde{P}=-\tilde{P}^\top$,
the superscript $\top$ denoting matrix transpose.
In addition to fixing the geometry of the hypothesis, we have the added benefit of reducing the parameter count from $54$ to $18$ while eliminating the degenerate pair~\eqref{eq:degenEulLax} from the hypothesis class. This does not rule out every possible rank deficient Lax pair, therefore, in addition, we add a penalty that forces every variable $\Omega_k$ to remain represented in both operators,
\begin{equation}\label{eq:DegenEulObj}
\mathcal{J}_{\rm full}
=
\prod_{k=1}^{3}
\left\| \frac{\partial \tilde L}{\partial \Omega_k} \right\|_F^2
\left\| \frac{\partial \tilde P}{\partial \Omega_k} \right\|_F^2,
\end{equation}
where $\|A\|_F^2 := \operatorname{tr}(A^\top A)$. The product structure is essential in that if any $\Omega_k$ is dropped from either operator, the corresponding factor vanishes and the reciprocal $\mathcal{J}_{\rm full}^{-1}$ blows up.

The full optimization problem becomes
\begin{equation}\label{eq:EulerFull}
\min_{\eta \in \mathbb{R}^N} \mathcal{J}_{\rm Euler}[\eta]
= \min_{\eta \in \mathbb{R}^N} \delta\,\mathcal{J}^{-1}_{\rm full}[\eta] + (1 - \delta)\,\mathcal{J}_{\rm loss}[\eta] + s\,\mathcal{S}[\eta],
\end{equation}
with $\delta \in (0, 1)$ balancing nondegeneracy against compatibility, $\mathcal{S}[\eta] = \|\eta\|_{\ell^1}$ the sparsity penalty, and $s > 0$ its strength. Empirically, both ingredients are necessary. We emphasize again that without the Frobenius product penalty, numerical results consistently collapse onto fake pairs with the same structural character as \eqref{eq:degenEulLax}, different in the specific variable dropped and coefficients present but identical in their uselessness. Meanwhile, without the skew-symmetry constraint layered on top of the Frobenius penalty, superfluous pairs begin to appear that are not genuinely skew and do not carry the $\mathfrak{so}(3)^*$ structure. Only when both are imposed does SILO reliably recover a pair with the same structure as the classical pair \eqref{eq:KnownEulPair}.

\subsection{From the discovered pair to the Manakov completion}

The Frobenius-regularized SILO search of the previous subsection recovers a skew-symmetric base pair $(L_*, P_*)$ with the structure of \eqref{eq:KnownEulPair}. This pair reproduces the Euler top but, as noted above, fails to generate the full Liouville hierarchy. The point of this subsection is that, given $(L_*, P_*)$, that are skew-symmetric, the spectral completion does not need to be accomplished by further numerics. The completion is the solution of a linear system posed inside the Lie algebra, and once the ansatz is fixed there is no residual numerical search. We make this explicit to mark the division of labor: SILO supplies the full-rank skew symmetric skeleton, and the spectral completion is pure linear algebra in the spirit of Manakov \cite{Manakov1976} and its semisimple generalization by Mishchenko and Fomenko \cite{MishchenkoFomenko1978}.

We review the method here for completeness. Posit the polynomial family linear in $\lambda$,
\begin{equation}\label{eq:lambda_completion}
L(\lambda) = L_* + \lambda L_1, \qquad P(\lambda) = P_* + \lambda P_1,
\end{equation}
with $L_1, P_1$ constant matrices to be determined. Requiring $\dot L(\lambda) = [L(\lambda), P(\lambda)]$ for all $\lambda$ and matching orders gives
\begin{align}
\lambda^0:\quad & \dot L_* = [L_*, P_*], \label{eq:order_zero}\\
\lambda^1:\quad & [L_*, P_1] + [L_1, P_*] = 0, \label{eq:order_one}\\
\lambda^2:\quad & [L_1, P_1] = 0. \label{eq:order_two}
\end{align}
Equation \eqref{eq:order_zero} holds by construction of $(L_*, P_*)$. The remaining two are a linear system in $(L_1, P_1)$, decoupled from $\Omega$ sampling. To extract the $\Omega$-independent content of \eqref{eq:order_one}, expand $L_*$ and $P_*$ in the standard basis $E_k$ of $\mathfrak{so}(3)$:
\[
L_*(\Omega) = \sum_{k=1}^3 I_k \Omega_k E_k, \qquad P_*(\Omega) = \sum_{k=1}^3 \Omega_k E_k.
\]
Substituting into \eqref{eq:order_one} and using linear independence of the $\Omega_k$, the coupled condition decouples into three matrix equations:
\begin{equation}\label{eq:decoupled_euler}
I_k \,[E_k, P_1] = [E_k, L_1], \qquad k = 1, 2, 3.
\end{equation}
Now, restrict \eqref{eq:decoupled_euler} to the diagonal commuting ansatz $L_1 = \operatorname{diag}(l_1, l_2, l_3)$, $P_1 = \operatorname{diag}(p_1, p_2, p_3)$, which automatically satisfies \eqref{eq:order_two}. A direct computation (using that $[E_k, \operatorname{diag}(d_1, d_2, d_3)]$ has a single off-diagonal entry $d_j - d_l$, where $(j, l)$ are the indices other than $k$) reduces \eqref{eq:decoupled_euler} to three scalar relations,
\[
I_k (p_j - p_l) = l_j - l_l, \qquad (k, j, l)\ \text{cyclic}.
\]

The order-$\lambda$ relations $I_k(p_j - p_l) = l_j - l_l$ are the only
constraints. Summing the three cyclic copies cancels the $l$'s and leaves a single
relation on $P_1$, solved by $p_k = \text{const}$ and by $p_k = I_k$, which span
its solution space; hence $p_k = a + b\,I_k$. Back-substituting gives
$l_j - l_l = b\,I_k(I_j - I_l)$, so $l_k = d - b\,I_j I_l = d - b\,I_1 I_2 I_3/I_k$.
The completion is therefore
\[
P_1 = \operatorname{diag}(a + b\,I_k), \qquad
L_1 = \operatorname{diag}\!\Big(d - b\,\tfrac{I_1 I_2 I_3}{I_k}\Big),
\]
with $a$ and $d$ the inert shifts $P_1 \to P_1 + aI$, $L_1 \to L_1 + dI$ and $b$ an
overall scale. Requiring the Hamiltonian in the spectrum fixes only $b \neq 0$: the order-$\lambda$
part of $\operatorname{tr}(L^3)$ is
$3\sum_k l_k(M_k^2 - \|M\|^2) = -6b\,I_1 I_2 I_3\,H$ modulo the Casimir. The
constants $a$ and $d$ are inert identity shifts and $b$ is an overall scale, so on
the three-dimensional top the completion is fixed only up to this gauge, and any
nonzero $b$ already certifies integrability.

The simplest choice makes the point. Take $a = d = 0$ and $b = 1$, so that
\[
P_1 = \operatorname{diag}(I_1, I_2, I_3), \qquad
L_1 = -I_1 I_2 I_3\,\operatorname{diag}(I_1^{-1}, I_2^{-1}, I_3^{-1}).
\]
As discussed, $\operatorname{tr}(L^2)$ returns the Casimir at order $\lambda^0$, and
$\operatorname{tr}(L^3)$ returns $-6 I_1 I_2 I_3\,H$ at order $\lambda^1$, so both
functionally independent integrals are present and the flow is Liouville-certified.
Note that this completion is not the Manakov $J$ given by Equation~\eqref{eq:ManakovJ} resulting from the parameters
$$
a = \tfrac12(I_1 + I_2 + I_3), \qquad b = -1, \qquad
d = \tfrac14(I_1 + I_2 + I_3)^2 - (I_1 I_2 + I_2 I_3 + I_3 I_1).
$$
Nothing internal to
the Euler top singles out $L_1 = J^2$, $P_1 = J$ over it: on $\mathfrak{so}(3)$ that
form is one representative of the gauge family $(a, b, d)$, distinguished only as
the $n = 3$ restriction of the $n$-dimensional rigid-body pair
\cite{Manakov1976, MishchenkoFomenko1978}, where $J$ is the physical mass matrix
and the Manakov shift/scale becomes forced.

\begin{remark} SILO's role is limited to the first step of finding meaningful Lax pairs, discovering
$(L_*, P_*)$ from the compatibility residual; the $\lambda$-graded completion is
then linear algebra over $\mathfrak{so}(3)$.
Extensions of SILO to higher-dimensional Lie-Poisson systems should be explored in future work, particularly within the rigid-body context. Two natural targets are the $n$-dimensional rigid body on $\mathfrak{so}(n)$~\cite{Manakov1976, MishchenkoFomenko1978}, whose spectral completion via the shift-of-argument method on semisimple Lie algebras tests the algebraic step of Section~\ref{sec:euler} at scale, and the Kowalewski top~\cite{BobenkoReymanSemenov1989}, whose Lax representation was found only in 1989 and is a natural stress test for any sparse-regression discovery procedure.
\end{remark}
\section{The Free Schr\"odinger Equation}
\label{sec:schrodinger}

We now move from the matrix setting to a scalar auxiliary formulation for PDEs, using the free Schr\"odinger equation,
\be
\label{e:LinearSchrodinger}
q_t=iq_{xx},
\ee
as a test case. The system is linear and classically solved via Fourier transforms. 
The associated Lax pair, in the Unified Transform Method (UTM) formalism \cite{Fokas1997,Fokas2008}, consists of two first-order equations for an auxiliary field depending on a spectral parameter $k$. The example is simple enough that every step can be done by hand, and yet it already exhibits a phenomenon that will recur: two different nondegeneracy penalties, applied to the same compatibility condition, return two structurally different Lax pairs, only one of which is classical.

\subsection{The Lax pair as a compatibility condition, learning ansatz and algebraic closure}

A Lax pair for \eqref{e:LinearSchrodinger} consists of two first-order equations for an auxiliary field $\mu(x,t;k)$:
\be
\mu_x = f(\mu,q;k),\qquad
\mu_t = g(\mu,q,q_x;k).
\label{e:LSLP}
\ee
This may be viewed as defining the Pfaffian one-form
\[
\omega = d\mu - f\,dx - g\,dt.
\]
Compatibility, $\mu_{xt}=\mu_{tx}$, is equivalent to Frobenius integrability of $\omega=0$. Indeed, along a graph $\mu=\mu(x,t;k)$ we have $d\mu=\mu_x\,dx+\mu_t\,dt$, and inserting this into $\omega=0$ reproduces the pair.  The one-form language is, of course, a compact way to encode the overdetermined first-order system and its compatibility, yet for our search purposes, the key relation is
\[
\partial_t f(\mu,q;k)=\partial_x g(\mu,q,q_x;k).
\]
This type of \emph{Clairaut compatibility} also applies for integrable nonlinear PDEs 
\cite{AblowitzSegur,NMPZ1984}, and will be further exploited in subsequent sections.

Now, the classical Lax pair corresponds to the choice
\be
f(\mu,q;k) = ik\,\mu + q,\qquad 
g(\mu,q,q_x;k) = -ik^2\,\mu + iq_x - kq.
\ee
To pose a search problem for~\eqref{e:LSLP}, we hypothesize an affine form for $f$ and $g$: 
\[
\tilde f = a_1(k)\,q + a_2(k)\,\mu,\qquad 
\tilde g = a_3(k)\,q + a_4(k)\,q_x + a_5(k)\,\mu,
\]
with coefficients at most quadratic in the spectral parameter $k$,
\[
a_j(k)=\sum_{l=0}^2 z_{lj}\,k^l,\qquad z_{lj}\in\mathbb{C},\ k\in\mathbb{R}.
\]
A direct computation of $\tilde f_t$ and $\tilde g_x$ under Schr\"odinger dynamics gives
\[
\tilde f_t = ia_1 q_{xx} + a_2(a_3 q + a_4 q_x + a_5\mu),\qquad
\tilde g_x = a_3 q_x + a_4 q_{xx} + a_5(a_1 q + a_2\mu).
\]
The Clairaut compatibility $\tilde f_t - \tilde g_x = 0$ reads
\[
\tilde f_t - \tilde g_x
= (ia_1-a_4)\,q_{xx}
+ (a_2 a_3 - a_5 a_1)\,q
+ (a_2 a_4 - a_3)\,q_x
+ (a_2 a_5 - a_5 a_2)\,\mu = 0,
\]
where the last term vanishes identically. The remaining coefficients yield three nonlinear algebraic constraints:
\begin{equation}\label{eq:SchrodAEq}
ia_1 - a_4 = 0,\qquad a_2 a_3 - a_5 a_1 = 0,\qquad a_2 a_4 - a_3 = 0,
\end{equation}
or more compactly, $F(\mathbf{a})=\mathbf{0}$. It is clear that the $k$-dependence can be encoded in a matrix,
\[
Z\,\mathbf{k}=\mathbf{a},\qquad
Z=\begin{pmatrix}
z_{01} & z_{02} & z_{03} & z_{04} & z_{05}\\
z_{11} & z_{12} & z_{13} & z_{14} & z_{15}\\
z_{21} & z_{22} & z_{23} & z_{24} & z_{25}
\end{pmatrix}^\top,\qquad
\mathbf{k}=\begin{pmatrix}1\\k\\k^2\end{pmatrix}.
\]
Rather than solve $F(\mathbf{a})=0$ symbolically, we relax the root-finding
problem into a sparse regression over the coefficient vector $z=\operatorname{vec}(Z)$:
\[
\min_{z\in\mathbb{C}^{15}}\;
\frac{1}{2}\,\mathbb{E}_{k\sim\rho}\!\left[\sum_{j=1}^{3}\|F_j(\mathbf{a})\|^2\right]
+ s\,\|z\|_{\ell^1(\mathbb{C}^{15})},
\qquad
\mathbf{a}=Z\mathbf{k},
\]
with $Z\in\mathbb{C}^{5\times 3}$ and $\rho$ uniform on $[-1,1]$. 
The numerics recover the sparsest nontrivial Lax pair
\[
\mu_x=iq,\qquad \mu_t=-q_x,
\]
which is irrelevant for UTM purposes: 
it contains no spectral parameter, and therefore 
it cannot be used to generate a meaningful transform in the UTM sense. 
This is the first nondegeneracy obstacle, and it is resolved by steering away from rank-deficient coefficient matrices.

Before doing so, we note that this type of degeneracy
is general for linear PDE. The pair $\mu_x = iq$, $\mu_t = -q_x$ is the $k=0$ reduction of the
Fokas unified-transform Lax pair \cite{Fokas1997, Fokas2008}, and every scalar
constant-coefficient linear evolution equation has one. For
$q_t + \omega(-i\partial_x)q = 0$ with polynomial dispersion $\omega$ obeying
$\omega(0)=0$, that pair is
\[
\mu_x - ik\mu = q, \qquad
\mu_t + \omega(k)\mu = i\,\frac{\omega(k)-\omega(-i\partial_x)}{k+i\partial_x}\,q,
\]
whose compatibility $\mu_{xt}=\mu_{tx}$ returns the equation for every $k$. At
$k=0$ it collapses to
\[
\mu_x = q, \qquad \mu_t = \partial_x^{-1}q_t = -\partial_x^{-1}\omega(-i\partial_x)\,q,
\]
which carries no spectral parameter and reproduces the equation as the single
conservation law $\mu_{xt}=\mu_{tx}$. The free Schr\"odinger case is
$\omega(k)=ik^2$, for which $\mu_t = \partial_x^{-1}(iq_{xx}) = iq_x$, and the
rescaling $\mu\mapsto i\mu$ returns the degenerate SILO pair $\mu_x = iq$, $\mu_t = -q_x$.

\subsection{Two nondegeneracy penalties and the resulting Lax pairs }

The classical Schr\"odinger Lax pair has coefficient matrix
\[
Z_{\text{classical}} =
\begin{pmatrix}
1 & 0 & 0 & i & 0\\
0 & i & -1 & 0 & 0\\
0 & 0 & 0 & 0 & -i
\end{pmatrix}^\top,
\]
which has full column rank. A direct way to penalize spectral collapse is thus to add a barrier against near-singular $Z^\dagger Z$. Two natural choices present themselves.

The \emph{determinant barrier} penalizes volume collapse of the parallelepiped spanned by the columns of $Z$,
\[
\min_{\mathbf{a}\in\mathcal{A}}\;
\frac{1-r}{2}\mathbb{E}_{k\sim\rho}\!\left[\sum_{j=1}^{3}\|F_j(\mathbf{a})\|^2\right]
+ \frac{r}{2}\,|\mathcal{D}|^{-2}
+ s\,\|z\|_{\ell^1(\mathbb{C}^{15})},
\qquad
\mathcal{D}=\det(Z^\dagger Z).
\]
The \emph{Frobenius pseudoinverse surrogate} penalizes the worst-conditioned singular direction,
\[
\min_{\mathbf{a}\in\mathcal{A}}\;
\frac{1-r}{2}\mathbb{E}_{k\sim\rho}\!\left[\sum_{j=1}^{3}\|F_j(\mathbf{a})\|^2\right]
+ \frac{r}{2}\,\tilde{\mathcal{D}}
+ s\,\|z\|_{\ell^1(\mathbb{C}^{15})},
\qquad
\tilde{\mathcal{D}}=\|Z^+\|_F^2=\operatorname{tr}\!\big((Z^\dagger Z)^{-1}\big).
\]
Geometrically, both penalties act on the Gram matrix $G=Z^\dagger Z$. Full column rank of $Z$ is equivalent to positive definiteness of $G$, or equivalently to all singular values $\sigma_1,\sigma_2,\sigma_3$ of $Z$ being nonzero. The determinant
\[
\det(Z^\dagger Z)=\sigma_1^2\sigma_2^2\sigma_3^2
\]
measures the squared volume of the parallelepiped spanned by the columns of $Z$, so a determinant barrier prevents collapse of this volume and hence penalizes loss of rank. By contrast,
\[
\|Z^+\|_F^2=\operatorname{tr}((Z^\dagger Z)^{-1})=\sigma_1^{-2}+\sigma_2^{-2}+\sigma_3^{-2}
\]
measures inverse conditioning: it blows up whenever any singular direction of $Z$ becomes small. The determinant barrier is naturally interpreted as a volume-collapse penalty, whereas the pseudoinverse penalty is a conditioning penalty that more strongly discourages thin or nearly collapsed directions. We note that both barriers can be weakened by rescaling $z$ to enlarge the singular values of $Z$; however, the $\ell^1$ term counteracts this by penalizing the overall size of the coefficients.

The two penalties select different pairs. The determinant penalty regularly returns the classical five-term Schr\"odinger Lax pair. The Frobenius surrogate, however, returns a different sparse solution, with coefficients of the form
\[
a_2(k)\equiv z_4,\qquad a_5(k)\equiv z_{11},
\]
and
\[
a_1(k)=z_1+z_2 k+z_3 k^2,\qquad
a_3(k)=z_5+z_6 k+z_7 k^2,\qquad
a_4(k)=z_8+z_9 k+z_{10} k^2,
\]
where the eleven complex parameters $z_j$ satisfy \eqref{eq:SchrodAEq}. This second pair is sparse, compatible, and not equivalent to the classical one by any straightforward gauge. It is, however, \emph{spectrally degenerate} in a precise algebraic sense, as a brief manual calculation shows.

For demonstration purposes, set $z_4=1$ and $z_{11}=i$. With this choice, Equations \eqref{eq:SchrodAEq} are satisfied for every $k$ by any function $M=M(k)$ through the relations $a_1(k)=M(k)$ and $a_3(k)=a_4(k)=iM(k)$. The resulting pair is
\begin{equation}\label{eq:SchrodWeakPair}
\mu_x = M(k)\,q + \mu,\qquad 
\mu_t = iM(k)(q+q_x) + i\mu,
\end{equation}
which admits the Schr\"odinger equation $q_t=iq_{xx}$ as its Clairaut compatibility condition for every choice of the function $M(k)$. The arbitrariness of $M(k)$ is structurally revealing: the pair carries no spectral information at all. Different $M(k)$ produce different auxiliary systems, all compatible with the same underlying PDE, with no way to select among them by the compatibility condition alone. In the language of the UTM, this family cannot produce a nontrivial transform kernel; the spectral parameter is present but acts as a free function rather than as the encoding of a scattering problem.

In fact, it is easy to see that pair given by Equation~\eqref{eq:SchrodWeakPair} is fake using the diagnostics of Appendix~\ref{app:butlerhay}. The pair is $g$-fake since
the gauge $\mu\mapsto M(k)^{-1}\mu$ above sends it to the $M(k)\equiv1$ pair,
removing $k$. It is $u$-fake: imposing \eqref{eq:SchrodAEq} on the generalized
coefficients with $a_2=z_4$, $a_5=z_{11}$ gives
\[
a_4=ia_1,\qquad a_3=iz_4\,a_1,\qquad z_{11}=iz_4^2,
\]
so $a_1=M(k)$ is unconstrained and the compatibility system is underdetermined.
Additionally, the Krichever test of Appendix~\ref{app:krichever} agrees: linearizing at
$q=q_0$, the auxiliary system supports the single mode $\mu\propto e^{x+it}$,
independent of $k$, so $\omega_{\mathrm{Lax}}$ is constant while
$\omega_{\mathrm{Schr}}(\xi)=\xi^2$.

The broader lesson of this section is that this is another example where compatibility does not select among Lax representations. Even after enforcing rank-fullness in coefficient space, the Clairaut condition leaves room for analytically unhelpful alternatives. Structural selection, i.e., determining which compatible pair actually encodes scattering, conservation, or inverse-problem data, is a separate question from compatibility itself, and must be addressed through additional constraints that reflect the intended use of the pair.

\section{The Inviscid Burgers Equation and Infinite Parametric Families of Lax Pairs}
\label{sec:burgers}

The Schr\"odinger example exhibited nonuniqueness in the form of a single spectrally degenerate alternative to the classical pair. In this section, we will show that the Hopf equation exhibits a different yet also striking form of nonuniqueness: 
an infinite parametric family of compatible auxiliary systems, indexed by a positive integer, all of which reduce to a base pair under Galilean transformations. The family emerged through a SILO search that was deliberately open about the exponent structure of the compatibility, and once one member was identified numerically, the full family could be verified by hand for every integer $m\geq 1$.

\subsection{A pen-and-paper auxiliary pair and SILO Formulation}

The inviscid Burgers (or Hopf) equation,
\be
u_t+u\,u_x=0,
\label{e:Hopf}
\ee
admits a WKB-type auxiliary formulation in which a scalar phase field $S(x,t)$ plays the role of the WKB phase. 
A pen-and-paper calculation of the dispersionless limit of the Kadomtsev-Petviashvili equation \cite{Zakharov1994} produces, as a special case, the pair
\be
S_x^2=:f(u;\lambda),\qquad
S_t=:g(u,S_x),
\label{e:HopfLP1}
\ee
with 
\be
f(u;\lambda) := -\tfrac{1}{6}u-\lambda,\qquad
g(u,S_x) :=-4S_x^3-u S_x\,,
\label{e:HopfLP2}
\ee
and with spectral parameter $\lambda\in\mathbb{C}$. This is a dispersionless Lax-type pair 
in the sense that compatibility of~\eqref{e:HopfLP1} reproduces the dynamics of $u$ through a Clairaut-type condition.

Before posing the search problem, it is useful to derive the compatibility conditions in a form that makes clear what we are enforcing. Write $s:=S_x$ and treat $f=f(u)$ and $g=g(u,s)$. Differentiating $f(u)=s^2$ in $x$ gives $2s\,s_x=f_u(u)\,u_x$, implying 
\[
s_x=\frac{f_u(u)}{2s}\,u_x.
\]
Differentiating $f(u)=s^2$ in $t$ and using \eqref{e:Hopf} yields
$2s\,s_t=f_u(u)\,u_t=-f_u(u)\,u\,u_x$, implying 
\[
s_t=-\frac{f_u(u)}{2s}\,u\,u_x.
\]
On the other hand, differentiating the second of \eqref{e:HopfLP1} in $x$ gives
\[
s_t=S_{xt}=(S_t)_x=g_u(u,s)\,u_x+g_s(u,s)\,s_x.
\]
Substituting the expression for $s_x$ and imposing $S_{xt}=S_{tx}$ yields, after factoring out $u_x$,
\[
g_u(u,s)+\frac{f_u(u)}{2s}\,g_s(u,s)=-\frac{f_u(u)}{2s}\,u,
\]
or, equivalently the compatibility identity
\begin{equation}\label{eq:BurgersClairaut}
2s\,g_u+f_u\bigl(g_s+u\bigr)=0.
\end{equation}
A direct substitution verifies that the pen-and-paper pair \eqref{e:HopfLP1}--\eqref{e:HopfLP2} satisfies \eqref{eq:BurgersClairaut} identically. This is the constraint that the regression enforces.

Because \eqref{eq:BurgersClairaut} involves only $u$ and $s$, we hypothesize $f$ and $g$ as polynomials in those variables,
\begin{equation}\label{eq:BurgersAnsatz}
\tilde f(u)=\sum_{j=0}^{N_1}\xi_j\,u^j,\qquad
\tilde g(s,u)=\sum_{j=0}^{N_2}\sum_{k=0}^{N_3}\zeta_{j,k}\,s^j u^k,
\end{equation}
and collect the unknowns into $\eta = (\Xi, \mathcal{Z}) \in \R^{N}$ with $N = (N_1+1) + (N_2+1)(N_3+1)$. Substituting \eqref{eq:BurgersAnsatz} into \eqref{eq:BurgersClairaut} produces a residual that is polynomial in $(u, s)$ and quadratic in $\eta$. SILO in this setting thus takes the form
\begin{equation}\label{eq:BurgersSILO}
\min_{\eta \in \R^N}\;
\mathbb{E}_{(u,s)\sim\nu}\!\left[
\frac{\bigl(2s\,\partial_u\tilde g + \partial_u\tilde f\,(\partial_s\tilde g + u)\bigr)^{2}}{\|\Xi\|_\infty\,\|\mathcal{Z}\|_\infty}
\right]
+ s\,\|\eta\|_{\ell^1},
\end{equation}
with $\nu$ uniform on $[-1, 1]^2$. The $\ell^\infty$ product normalization in the denominator plays the role of the Frobenius product penalty used in Section~\ref{sec:euler}: it diverges whenever either $\Xi$ or $\mathcal{Z}$ collapses to zero, and so suppresses the trivial solutions $\tilde f \equiv 0$ and $\tilde g \equiv 0$. The $\ell^1$ term once again promotes parsimony.

\subsection{Three numerically discovered families}
\label{sec:burgers-three-families}

Three sparse families emerge from the numerics. The first is the pen-and-paper pair. The second, discovered by SILO, may be written (after rescaling) as
\be
S_x^2 = a u + \lambda, \qquad
S_t = b u - u S_x - \frac{b}{a} S_x^2 + \frac{2}{3a} S_x^3,
\ee
with $a \neq 0$ and free constants $(a, b, \lambda)$. A third, higher-degree family also appears:
\be
S_x^2 = a u + \lambda, \qquad
S_t = c_1 S_x^4 + c_2 S_x^3 + c_3 S_x^2 + c_4 u S_x^2 + c_5 u S_x + c_6 u + c_7 u^2,
\label{e:HopfLP3}
\ee
with the $c_j$ satisfying linear relations derived below.

It is clear that even higher-degree families would emerge from the numerics if we include more terms in the polynomial library.  Rather than verify each family independently or proceed by induction, we seek to solve the compatibility PDE~\eqref{eq:BurgersClairaut} in closed form. 
Doing so reveals that the pattern is more general than the two polynomial samples above suggest: there is an infinite-dimensional family of admissible temporal parts $S_t$, indexed by an arbitrary function of one variable, every element of which collapses to the base pair on the constraint surface. The second and third numerical families are the linear and quadratic polynomial samples, respectively.

\paragraph{The compatibility PDE and its general solution.}
Fix the spatial ansatz $S_x^2 = a u + \lambda$ with $a \neq 0$, and, again, write $s := S_x$ for brevity. With $f(u) = a u + \lambda$, the Clairaut identity \eqref{eq:BurgersClairaut} becomes the linear first-order PDE
\begin{equation}\label{eq:BurgersCompatPDE}
a\, g_s + 2 s\, g_u = -a u
\end{equation}
for the unknown $g(s, u) = S_t$.  Solving~\eqref{eq:BurgersCompatPDE} is thus a routine application of the method of characteristics with  characteristic curves that satisfy $ds/a = du/(2s)$, or equivalently,
\(
d(s^2 - a u) = 0
\),
and therefore coincide with the level sets of the spatial constraint: the characteristic variable is $\xi := s^2 - a u$, and the surface $\xi = \lambda$ is the one on which the auxiliary pair is evaluated.

\begin{proposition}\label{prop:BurgersInfiniteFamily}
The general $C^1$ solution of \eqref{eq:BurgersCompatPDE} is
\begin{equation}\label{eq:BurgersGenSol}
g(s, u) = \frac{2}{3a}\, s^3 - u s + \Phi(s^2 - a u),
\end{equation}
where $\Phi\colon \mathbb{R} \to \mathbb{R}$ is arbitrary. 
Every Lax pair for the inviscid Burgers equation with spatial ansatz $S_x^2 = a u + \lambda$ has temporal part of this form.
Moreover, on the constraint surface $s^2 = a u + \lambda$, every element of the family \eqref{eq:BurgersGenSol} reduces to
\vspace*{-0.5ex}
\begin{equation}
\label{e:S_translate}
S_t = -\frac{S_x^3}{3a} + \frac{\lambda}{a}\, S_x + \Phi(\lambda),
\end{equation}
a Galilean translate of the base pair with additive constant $\mu := \Phi(\lambda)$.
\end{proposition}

\begin{proof}
The function $g_0(s, u) := {2s^3}/({3a})  - u s$ is a particular solution 
of \eqref{eq:BurgersCompatPDE}, since $(g_0)_s = 2 s^2 / a - u$ and $(g_0)_u = -s$ imply 
\[
a\, (g_0)_s + 2 s\, (g_0)_u = 2 s^2 - a u - 2 s^2 = -a u.
\]
The associated homogeneous equation $a h_s + 2 s h_u = 0$ has general solution $h = \Phi(\xi)$ with $\xi = s^2 - a u$, since $\xi$ is constant along characteristics. Summing gives \eqref{eq:BurgersGenSol}.
Equation~\eqref{e:S_translate} then follows by
simply substituting $u = (s^2 - \lambda)/a$ into \eqref{eq:BurgersGenSol}, which yields
\[
g = \frac{2}{3a} s^3 - \frac{s^2 - \lambda}{a}\, s + \Phi(\lambda) = -\frac{s^3}{3a} + \frac{\lambda s}{a} + \Phi(\lambda). 
\qedhere
\]
\end{proof}

\paragraph{The two numerical families as polynomial samples.}
The second and third families correspond to linear and quadratic choices of $\Phi$, respectively. Taking $\Phi(\xi) = -b \xi / a$ in \eqref{eq:BurgersGenSol} gives
\[
g(s, u) = \frac{2}{3a}\, s^3 - u s - \frac{b}{a}(s^2 - a u) = \frac{2}{3a}\, S_x^3 - u S_x - \frac{b}{a}\, S_x^2 + b u,
\]
the second family. Taking $\Phi(\xi) = \alpha_0 + \alpha_1 \xi + \alpha_2 \xi^2$ and expanding yields
\[
g(s, u) = \frac{2}{3a}\, s^3 - u s + \alpha_0 + \alpha_1 (s^2 - a u) + \alpha_2 (s^2 - a u)^2.
\]
Matching monomial coefficients against the third family's basis $\{s^4, s^3, s^2, u s^2, u s, u, u^2, 1\}$ gives
\[
c_1 = \alpha_2,\quad c_2 = \tfrac{2}{3a},\quad c_3 = \alpha_1,\quad c_4 = -2 a \alpha_2,\quad c_5 = -1,\quad c_6 = -a \alpha_1,\quad c_7 = a^2 \alpha_2,
\]
with additive constant $\alpha_0$. The seven $c_j$ are thus determined by three free parameters $(\alpha_0, \alpha_1, \alpha_2)$, with $c_2$ and $c_5$ pinned by compatibility and $(c_1, c_4, c_7)$, $(c_3, c_6)$ related through $\alpha_2, \alpha_1$ respectively. The relations among the $c_j$ reported by the numerics are exactly these.

\paragraph{The infinite family.}
Proposition~\ref{prop:BurgersInfiniteFamily} says that the temporal part $S_t$ is
determined only up to an additive homogeneous solution $\Phi(\xi)$, $\xi = s^2 - au$,
a function constant along the characteristics of the compatibility PDE. As
solutions of that PDE on the full $(s,u)$ plane these are genuinely distinct, and
the solution space is infinite-dimensional. As Lax pairs they are not, that is, the
auxiliary system lives only on the constraint surface $\xi = \lambda$, where
$\Phi(\xi)$ collapses to the single constant $\Phi(\lambda)$, and that constant is
the potential gauge $S \mapsto S + \Phi(\lambda)\,t$, which fixes $S_x$ and the
recovered equation. Every member therefore reduces on the surface to the base pair
up to this gauge, as the Corollary records. Polynomial choices of $\Phi$ of degree
$N$ give $(s,u)$-polynomial $S_t$ of bidegree up to $(2N, N)$, and the sequence of
``new'' families SILO returns for $N = 2, 3, \ldots$ is the artifact of evaluating
$\Phi$ before restricting to $\xi = \lambda$. There is no ceiling to $N$, and no new
Lax pair past the base one.

\subsection{An infinite family of pairs}

 In light of
the above results, we applied the same approach assuming $S_x^3=f(u;\lambda)$. Once again SILO returned valid Lax pairs, but this time with a noticeably different algebraic structure. 
Pushing further, while investigating hypotheses of the form $S_x^m=f(u;\lambda)$, SILO discovered the following pair,
\begin{equation}\label{eq:BurgersFamily}
S_x^m=a\,u+\lambda,\qquad
S_t=-\frac{1}{a(m+1)}S_x^{m+1}+\frac{\lambda}{a}S_x+\mu,
\end{equation}
which is valid for every integer $m>0$, with $a\neq 0$ and arbitrary constants $(\lambda,\mu)$.

Showing this by hand for general $m$ is not difficult. Set $s:=S_x$ and assume
$s^m=a\,u+\lambda$, so that $u= (s^m-\lambda)/{a}$.
Differentiating in $x$ yields
\(m s^{m-1}s_x=a\,u_x\), i.e., 
\(u_x=\frac{m}{a}s^{m-1}s_x\).
Differentiating in $t$ and using inviscid Burgers yields instead
\[
m s^{m-1}s_t=a\,u_t=-a\,u\,u_x=-a\Big(\frac{s^m-\lambda}{a}\Big)\Big(\frac{m}{a}s^{m-1}s_x\Big).
\]
Canceling the common factor $m s^{m-1}$,
\[
s_t=-\frac{1}{a}(s^m-\lambda)s_x=\frac{1}{a}(\lambda-s^m)s_x.
\]
If we now seek $S_t=P(s)$, then $s_t=S_{xt}=(S_t)_x=P'(s)s_x$, 
so that 
$P'(s)=(\lambda-s^m)/a$ implies 
\[
P(s)=\frac{1}{a}\Big(\lambda s-\frac{s^{m+1}}{m+1}\Big)+\mu,
\]
which is exactly \eqref{eq:BurgersFamily}. The $m=2$ case coincides with the pair reported earlier in this section.

To interpret this result it is useful to recall Calogero and Nucci's key observation~\cite{CalogeroNucci} that any PDE with a local conservation law admits a Lax pair built by hodograph dressing of a constant-coefficient linear pair. We can therefore ask whether the inviscid Burgers family \eqref{eq:BurgersFamily} could also be derived from a conservation-law hodograph in the dispersionless WKB setting. The answer is that indeed it can, as we now show.

The starting observation is that inviscid Burgers admits a fractional-power conservation hierarchy. For any real $m > 0$,
\[
\partial_t u^{1/m} = \frac{1}{m} u^{1/m-1} u_t = -\frac{1}{m} u^{1/m} u_x = -\frac{1}{m+1}\partial_x u^{(m+1)/m},
\]
so $\partial_t u^{1/m} + \frac{1}{m+1}\partial_x u^{(m+1)/m} = 0$ is a conservation law of inviscid Burgers for every $m > 0$. Allowing a Galilean shift, set $v := au + \lambda$ for constants $a \neq 0$ and $\lambda$.  The inviscid Burgers in $u$ becomes $v_t = -(v - \lambda) v_x / a$ in $v$. A direct computation then gives
\[
\partial_t v^{1/m} = \frac{1}{m} v^{1/m - 1} v_t = -\frac{(v - \lambda) v_x}{m a} v^{1/m - 1} = \frac{\lambda v^{1/m - 1} v_x}{m a} - \frac{v^{1/m} v_x}{m a},
\]
and each term on the right is a total $x$-derivative:
\[
-\frac{v^{1/m} v_x}{m a} = -\frac{1}{a(m+1)} \partial_x v^{(m+1)/m}, \qquad \frac{\lambda v^{1/m - 1} v_x}{m a} = \frac{\lambda}{a} \partial_x v^{1/m}.
\]
Combining,
\[
\partial_t v^{1/m} = \partial_x\!\left[\frac{\lambda}{a} v^{1/m}-\frac{1}{a(m+1)} v^{(m+1)/m} \right],
\]
which is a one-parameter family of conservation laws indexed by $m>0$, 
with densities
$f_m$ and fluxes $g_m$ given by
\[
f_m:= (au + \lambda)^{1/m},\qquad 
g_m := \frac{\lambda}{a}(au + \lambda)^{1/m}-\frac{1}{a(m+1)}(au + \lambda)^{(m+1)/m} .
\]

With these conservation laws in hand, we now apply the hodograph construction. Recall that a local conservation law $\partial_t f = \partial_x g$ is the closedness condition for the one-form $\omega = f\,dx + g\,dt$, so on a simply connected solution domain there exists a scalar potential $S(x,t)$ with $\omega = dS$, giving $S_x = f$ and $S_t = g$. The conservation law of the underlying PDE is thereby promoted to a first-order auxiliary system in $S$. Apply this with $f = f_m$ and $g = g_m$ from above by setting
\[
S_x = f_m = (au+\lambda)^{1/m}, \qquad S_t = g_m + \mu,
\]
where the free parameter $\mu$ is the additive potential-gauge freedom $S \mapsto S + \mu t$, which leaves $S_x$ invariant and shifts $S_t$. The mixed-partial consistency $S_{xt} = S_{tx}$ reads $\partial_t f_m = \partial_x g_m$, exactly the conservation law derived in the previous paragraph, so $S$ exists along solutions of inviscid Burgers. 

Lastly, to recover \eqref{eq:BurgersFamily}, simply raise $S_x$ to the $m$-th power, and we have immediately the spatial part $S_x^m = au + \lambda$. For the temporal part, note that $S_x^{m+1} = (au+\lambda)^{(m+1)/m}$, so each term in $g_m$ rewrites cleanly as a power of $S_x$:
\[
g_m = \frac{\lambda}{a}\,S_x-\frac{1}{a(m+1)}\,S_x^{m+1} .
\]
Substituting into $S_t = g_m + \mu$ gives the temporal equation of \eqref{eq:BurgersFamily}.

The family that SILO discovered is therefore a dispersionless first-order incarnation of the Calogero--Nucci mechanism, with the integer $m>0$ indexing the fractional-power root of the underlying conservation law and the constants $(a, \lambda, \mu)$ playing the roles of scaling, Galilean shift, and gauge. The integer indexing is itself an artifact of our polynomial library. Indeed, the derivation above goes through for any real $m > 0$, since the conservation law $\partial_t v^{1/m} = \partial_x[\,\cdots\,]$ requires nothing of $m$ beyond $m \neq 0$, and the spatial relation $S_x^m = au + \lambda$ is a well-defined algebraic constraint on any branch where $au + \lambda > 0$. Inviscid Burgers therefore admits a genuine continuum of compatible dispersionless auxiliary systems indexed by $m \in \mathbb{R}_{>0}$. This is the inviscid Burgers incarnation of ``Lax pairs galore'' in its strongest form, with uncountably many compatibility certificates.

\section{Shallow Water Equations and Dispersionless Lax Functions}
\label{sec:shallowwater}

Shallow water extends the inviscid Burgers analysis of section~\ref{sec:burgers} to two field components, and the Lax landscape splits into two registers. Within the Laurent layer of dispersionless Lax functions in a spectral parameter $p$, the classical Brunelli--Das pair~\cite{BrunelliDas1997} is the only object our sparse search returns across the Laurent 
and polynomial degrees we explored. Outside the Laurent layer, the Calogero--Nucci hodograph mechanism produces an infinite-dimensional family of compatible auxiliary systems, parametrized by solutions of an Euler--Poisson--Darboux equation in the Riemann invariants. The Brunelli--Das pair sits inside this continuum as the algebraic family generated by the spectral discriminant.

\subsection{The Brunelli--Das dispersionless Lax representation}
\label{sec:sw-bd}

The shallow water system (also known as the Saint-Venant system, which is equivalent to the polytropic gas equations with a pressure term $P = \pm\frac12 a \rho^\gamma$ with $\gamma=2$) reads 
\begin{equation}\label{eq:SW-pde}
\rho_t + (u\rho)_x = 0, \qquad u_t + u u_x + a\rho_x = 0, \qquad a = \pm 1,
\end{equation}
and arises as the dispersionless (semiclassical) limit of focusing and
defocusing NLS equations in its Madelung (i.e., density and phase gradient/velocity) representation~\cite{JinLevermoreMcLaughlin1999}. 
Brunelli and Das~\cite{BrunelliDas1997} gave a Lax description of the polytropic gas
hierarchy;
here we limit our considerations to the case $\gamma=2$. 
Their result is the dispersionless Lax function
\begin{equation}\label{eq:BD-L}
L(p; u, \rho) = p + u + \frac{a\rho}{p},
\end{equation}
together with the polynomial generator
\begin{equation}\label{eq:BD-B}
B(p; u) = \tfrac{1}{2} p^2 + u p,
\end{equation}
related by the standard hierarchy projection
\begin{equation}\label{eq:BD-projection}
B = \bigl( \tfrac{1}{2} L^2 \bigr)_{>0},
\end{equation}
where $(\,\cdot\,)_{>0}$ keeps strictly positive powers of $p$. The
dispersionless Lax equation
\begin{equation}\label{eq:BD-Lax}
L_t = \{L, B\}, \qquad \{L, B\} := L_p B_x - L_x B_p,
\end{equation}
in which $p$ is an independent spectral variable, so $L_t$, $L_x$, $B_x$ act through $u(x,t),\rho(x,t)$ at fixed $p$ and $L_p$, $B_p$ are the explicit $p$-derivatives, reproduces \eqref{eq:SW-pde} after
matching coefficients of $p^0$ and $p^{-1}$. We refer the reader to
\cite{BrunelliDas1997, Brunelli2000, ConstandacheDasToppan2002} for the
algebraic structure (recursion relations, biHamiltonian content, and
generating functions for the conserved charges), and to \cite{KonopelchenkoMartinezMedina2010} for a treatment of dispersionless coupled-KdV hodographs in which the Euler--Poisson--Darboux equation that we encounter below appears as the
compatibility datum for hierarchy critical points.

Similarly as in Section~\ref{sec:burgers}, the Lax function \eqref{eq:BD-L} admits an equivalent phase-function form: 
 introducing a phase $S(x, t; \lambda)$ such that $p := S_x$, the spectral curve
$L = \lambda$ becomes the quadratic relation
\begin{equation}\label{eq:SW-quad}
p^2 + (u - \lambda) p + a\rho = 0,
\end{equation}
and \eqref{eq:BD-Lax} is equivalent, on the level set $L = \lambda$, to the
Clairaut compatibility $S_{xt} = S_{tx}$ for $S_x = p$ and $S_t = -B$. We
will return to this form in Section~\ref{sec:sw-cornucopia} with more detail when
we generate alternative pairs by the hodograph mechanism.

\subsection{SILO formulation and recovery of the known pair}
\label{sec:sw-silo}

The hierarchy projection \eqref{eq:BD-projection} is the structural feature
that makes a SILO formulation fairly straightforward: the
time generator $B$ is not an independent unknown, and the inverse problem
reduces to learning a single object $L$. 
We therefore posit for $L$ 
the truncated Laurent expansion 
\begin{equation}\label{eq:SW-ansatz}
\widetilde L(p; u, \rho) = \sum_{m = -M_-}^{M_+} c_m(u, \rho)\, p^m, \qquad
c_m(u, \rho) = \sum_{j=0}^{N_u} \sum_{k=0}^{N_\rho} \theta^{(m)}_{j,k}\, u^j \rho^k,
\end{equation}
with normalization $c_1 \equiv 1$ to fix the spectral scaling. 
To derive the regression residuals, we expand $L_t = L_u u_t + L_\rho \rho_t$ and
substitute the dynamical equations~\eqref{eq:SW-pde}, obtaining 
\[
L_t = \bigl(-u L_u - \rho L_\rho\bigr) u_x + \bigl(-a L_u - u L_\rho\bigr) \rho_x.
\]
For the Poisson bracket $\{L, B\} = L_p B_x - L_x B_p$, the spatial derivatives $L_x$ and $B_x$ at fixed $p$ expand via the chain rule as
\[
L_x = L_u u_x + L_\rho \rho_x, \qquad B_x = B_u u_x + B_\rho \rho_x.
\]
Substituting and grouping by $u_x$ and $\rho_x$,
\begin{align*}
\{L, B\} = L_p B_x - L_x B_p = L_p (B_u u_x + B_\rho \rho_x) - (L_u u_x + L_\rho \rho_x) B_p 
= (L_p B_u - L_u B_p) u_x + (L_p B_\rho - L_\rho B_p) \rho_x.
\end{align*}
Equating coefficients gives two pointwise residuals
\bse
\begin{align}
\mathcal{R}_1(\widetilde L) &= -u \widetilde L_u - \rho \widetilde L_\rho - \bigl(\widetilde L_p \widetilde B_u - \widetilde L_u \widetilde B_p\bigr), \label{eq:SW-R1}\\
\mathcal{R}_2(\widetilde L) &= -a \widetilde L_u - u \widetilde L_\rho - \bigl(\widetilde L_p \widetilde B_\rho - \widetilde L_\rho \widetilde B_p\bigr), \label{eq:SW-R2}
\end{align}
\ese
with $\widetilde B = (\tfrac{1}{2} \widetilde L^2)_{>0}$, leading to the SILO problem
\begin{equation}\label{eq:SW-SILO}
\min_{\theta \in \Theta}\; (1-r)\, \mathbb{E}_{(u, \rho, p) \sim \nu}\Bigl[|\mathcal R_1|^2 + |\mathcal R_2|^2\Bigr] + r\, \mathcal N(\widetilde L) + 2s\, \|\theta\|_{\ell^1},
\end{equation}
where $\nu$ is supported on a bounded set with $\rho > 0$ and
$|p| \geq p_{\min} > 0$. The penalty
\begin{equation}\label{eq:SW-N}
\mathcal N(\widetilde L) = \frac{1}{\mathbb{E}_{(u,\rho,p)\sim\nu}\!\bigl[\widetilde L_u^2\,\widetilde L_\rho^2\bigr]}
\end{equation}
is the shallow-water analog of the Frobenius product penalty
\eqref{eq:DegenEulObj} used for the Euler top: it diverges whenever
$\widetilde L_u$ or $\widetilde L_\rho$ vanishes identically on the support
of $\nu$, forcing both fields to remain represented in $\widetilde L$.
Without it, the optimizer collapses onto $\widetilde L = p$, which
satisfies \eqref{eq:SW-R1}--\eqref{eq:SW-R2} identically but carries no
field content.

In the experiment we take $M_- = 2$, $M_+ = 1$, and $N_u = N_\rho = 2$, so the
classical three-term ansatz is not built in a priori. Across random
initializations the optimizer returns the Brunelli--Das pair \eqref{eq:BD-L}
and its constant-shift gauges $L \mapsto L + c$, $B \mapsto B + c\, p$,
and reports no other compatible solution within the ansatz.
Within the Laurent class, the known Lax function is the only object SILO
could find, even after enlarging the search window to $M_-, M_+ \leq 5$
and the polynomial bidegrees to $N_u, N_\rho \leq 3$. Outside the Laurent class, however, the picture is entirely
different, as we show next.

\subsection{A Calogero--Nucci cornucopia of pairs via the Euler--Poisson--Darboux equation}
\label{sec:sw-cornucopia}

Recall that section~\ref{sec:burgers} produced an infinite family of Lax
representations for inviscid Burgers via the Calogero--Nucci hodograph
mechanism. In the present case, a SILO sweep on a hodograph ansatz would, by same token, surface a handful of low-degree members of
the family one at a time and leave the structural pattern to be
recognized afterward by hand. 
Here, we skip
the relevant numerics. 
The mechanism applies verbatim to the shallow water equations, but with a resulting family that is substantially richer than for the Hopf equation, because the shallow water equations have two Riemann invariants instead of just one.
Note that in this subsection we set $a = +1$ for definiteness.
The case $a = -1$ only requires a sign tracking throughout.

\paragraph{Riemann coordinates and the compatibility equation.}

The Riemann invariants of \eqref{eq:SW-pde} are
\begin{equation}\label{eq:SW-Riemann}
r_+ = u + 2 \sqrt{\rho}, \qquad r_- = u - 2 \sqrt{\rho},
\end{equation}
and a direct calculation gives the Riemann form
\begin{equation}\label{eq:SW-Riemann-form}
(r_\pm)_t + c_\pm\, (r_\pm)_x = 0, \qquad c_\pm = u \pm \sqrt{\rho} = \frac{3 r_\pm + r_\mp}{4},
\end{equation}
two coupled Hopf-like equations whose characteristic speeds depend on
\emph{both} invariants. Unlike scalar Hopf, therefore, no single-variable function $f(r_+)$ alone is conserved.  Thus, the analog of the fractional-power trick of
Section~\ref{sec:burgers} requires a joint density $F(r_+, r_-)$.

A Calogero--Nucci hodograph pair takes the form $S_x = F(u, \rho)$,
$S_t = G(u, \rho)$, with compatibility $S_{xt} = S_{tx}$ equivalent on
solutions of \eqref{eq:SW-pde} to the conservation law
\begin{equation}\label{eq:SW-cons-law}
\partial_t F = \partial_x G.
\end{equation}
Working in Riemann coordinates and writing $F = F(r_+, r_-)$, $G = G(r_+, r_-)$, the chain rule together with the Riemann form $(r_\pm)_t = -c_\pm (r_\pm)_x$ gives
\begin{align*}
\partial_t F &= F_{r_+} (r_+)_t + F_{r_-} (r_-)_t = -c_+ F_{r_+} (r_+)_x - c_- F_{r_-} (r_-)_x, \\
\partial_x G &= G_{r_+} (r_+)_x + G_{r_-} (r_-)_x.
\end{align*}
Matching coefficients of $(r_+)_x$ and $(r_-)_x$ in $\partial_t F = \partial_x G$ forces the flux to satisfy
\begin{equation}\label{eq:SW-flux}
G_{r_+} = -c_+ F_{r_+}, \qquad G_{r_-} = -c_- F_{r_-}.
\end{equation}
Differentiating these, and using $(c_+)_{r_-} = (c_-)_{r_+} = 1/4$ from $c_\pm = (3 r_\pm + r_\mp)/4$,
\[
G_{r_+ r_-} = -\tfrac{1}{4} F_{r_+} - c_+ F_{r_+ r_-}, \qquad G_{r_- r_+} = -\tfrac{1}{4} F_{r_-} - c_- F_{r_+ r_-}.
\]
Setting $G_{r_+ r_-} = G_{r_- r_+}$ and using $c_+ - c_- = (r_+ - r_-)/2$ reduces to the linear second-order PDE
\begin{equation}\label{eq:EPD}
\tfrac12(r_+ - r_-)\, F_{r_+ r_-} + \tfrac{1}{4}\bigl( F_{r_+} - F_{r_-} \bigr) = 0.
\end{equation}

This is a classical Euler--Poisson--Darboux (EPD) equation
\cite{Whitham1974, CourantHilbert}, the same
equation that organizes the Carrier--Greenspan transformation of long-wave
runup \cite{CarrierGreenspan1958} and the hydrodynamic-type integrability
theory of Tsarev \cite{Tsarev1991}, and that arises in
\cite{KonopelchenkoMartinezMedina2010} as the equation governing
hodograph critical points of the dispersionless coupled-KdV hierarchies. The solution space of \eqref{eq:EPD} is infinite-dimensional, and \emph{every} solution
gives a Lax pair for shallow water through \eqref{eq:SW-flux} and the
hodograph identification $S_x = F$, $S_t = G$. 
Here we exhibit three families explicitly.

\paragraph{The polynomial Whitham hierarchy.}

Restricting \eqref{eq:EPD} to polynomial $F(r_+, r_-)$ and solving by total degree, the basis grows by exactly one new direction per degree. The reason is dimensional since the space of homogeneous polynomials of degree $d$ in $(r_+, r_-)$ is $(d+1)$-dimensional, while EPD applied to a homogeneous degree-$d$ polynomial returns a homogeneous expression of degree $d-1$ whose vanishing imposes $d$ linear constraints, one for each monomial $r_+^i r_-^{d-1-i}$ with $i = 0, \ldots, d-1$. This leaves one free direction at each degree.

Degree $\leq 1$ is trivial: $F = a + b(r_+ + r_-)$ has $F_{r_+ r_-} = 0$ and $F_{r_+} - F_{r_-} = 0$, so EPD reduces to $0 = 0$ for any constants $a, b$. At degree $2$, write $F = \alpha r_+^2 + \beta r_+ r_- + \gamma r_-^2$. Then $F_{r_+ r_-} = \beta$ and $F_{r_+} - F_{r_-} = (2\alpha - \beta) r_+ + (\beta - 2\gamma) r_-$. Substituting into EPD,
\[
\tfrac{1}{2}\beta(r_+ - r_-) + \tfrac{1}{4}\bigl[(2\alpha - \beta) r_+ + (\beta - 2\gamma) r_-\bigr] = 0,
\]
and matching coefficients of $r_+$ and $r_-$ gives $\beta = -2\alpha$ and $\beta = -2\gamma$, so $\alpha = \gamma$. Setting $\alpha = 1$, the new direction is $F_2 = r_+^2 - 2 r_+ r_- + r_-^2 = (r_+ - r_-)^2$. The degree-$3$ calculation is analogous and forces the homogeneous cubic part to be proportional to $(r_+ + r_-)(r_+ - r_-)^2$. The pattern continues, giving a basis $\{F_0, F_1, F_2, \ldots\}$ of polynomial EPD solutions,
\[
F_0 = 1,\quad F_1 = r_+ + r_-,\quad F_2 = (r_+ - r_-)^2,\quad F_3 = (r_+ + r_-)(r_+ - r_-)^2,\quad \ldots,
\]
where the subspace of solutions of total degree at most $d$ has dimension $d + 1$.

Translating back to the original physical variables $(u, \rho)$ via $r_+ + r_- = 2u$ and $(r_+ - r_-)^2 = 16\rho$, the first six generators are, up to overall normalization,
\begin{equation}\label{eq:Whitham-generators}
1,\quad u,\quad \rho,\quad u\rho,\quad \rho^2 + u^2\rho,\quad u\rho(3\rho + u^2),\quad \ldots.
\end{equation}
We recognize this pattern as the well-known Kupershmidt--Manin / Whitham hierarchy of conserved densities for shallow water \cite{KupershmidtManin1977, Whitham1974}.

For each density $F$, the matching flux $G$ is determined directly by the conservation law $\partial_t F = \partial_x G$ on shallow-water solutions. Two examples make this concrete. The mass density $F = \rho$ uses $\rho_t = -(u\rho)_x$ from \eqref{eq:SW-pde}:
\[
\partial_t \rho = -(u\rho)_x = \partial_x(-u\rho), \qquad G = -u\rho,
\]
giving the hodograph pair $(S_x, S_t) = (\rho, -u\rho)$. The momentum-energy density $F = u\rho$ uses both PDEs and the product rule:
\begin{align*}
\partial_t(u\rho) &= u_t\,\rho + u\,\rho_t 
    = (-uu_x - \rho_x)\rho - u(u\rho)_x 
    = -2u\rho\, u_x - \rho\rho_x - u^2\rho_x 
    = -\partial_x\!\bigl(u^2\rho + \tfrac{1}{2}\rho^2\bigr),
\end{align*}
giving the hodograph pair $(S_x, S_t) = (u\rho, -u^2\rho - \tfrac{1}{2}\rho^2)$.

None of these is a Laurent function in any spectral parameter $\lambda$; they are polynomial in the field variables, and they are invisible to the Brunelli--Das presentation because that presentation parameterizes the spectral curve $L = \lambda$, not the conserved densities directly.

\paragraph{The algebraic continuum via the spectral discriminant.}

A second family of EPD solutions comes from the algebraic ansatz
\[
F(r_+, r_-; \lambda) = (r_+ - \lambda)^a (r_- - \lambda)^b
\]
with constants $a, b \in \R$ and spectral parameter $\lambda \in \C$. The derivatives are
\[
F_{r_+ r_-} = ab(r_+ - \lambda)^{a-1}(r_- - \lambda)^{b-1}, \qquad F_{r_+} - F_{r_-} = a(r_+ - \lambda)^{a-1}(r_- - \lambda)^b - b(r_+ - \lambda)^a (r_- - \lambda)^{b-1}.
\]
Substituting into EPD and factoring out $(r_+ - \lambda)^{a-1}(r_- - \lambda)^{b-1}$ leaves
\[
\tfrac12{ab}(r_+ - r_-) + \tfrac{1}{4}\bigl[a(r_- - \lambda) - b(r_+ - \lambda)\bigr] = 0,
\]
which expanded becomes $b(2a-1)\,r_+ + a(1-2b)\,r_- + (b - a)\lambda = 0$. Matching coefficients of $r_+, r_-, \lambda$ separately forces $a = b = 1/2$ as the unique nontrivial product solution. Hence
\begin{equation}\label{eq:F-cont}
F_\lambda(r_+, r_-) = \sqrt{(\lambda - r_+)(\lambda - r_-)}.
\end{equation}

To find the matching flux $G_\lambda$, use the relations \eqref{eq:SW-flux}:
\[
G_{r_+} = -c_+ F_{\lambda, r_+} 
    = \frac{(3r_+ + r_-)(\lambda - r_-)}{8 F_\lambda}.
\]
Take the ansatz $G_\lambda = h(r_+, r_-, \lambda)\, F_\lambda$. Then $G_{r_+} = h_{r_+} F_\lambda + h F_{\lambda, r_+}$. Multiplying both sides of the constraint by $F_\lambda$, using $F_\lambda^2 = (\lambda - r_+)(\lambda - r_-)$, and dividing through by $(\lambda - r_-)$ gives the reduced equation
\[
h_{r_+}(\lambda - r_+) = \tfrac18(3r_+ + r_-) + \tfrac12{h},
\]
with the symmetric companion $h_{r_-}(\lambda - r_-) = \tfrac18(r_+ + 3r_-) + \tfrac12{h}$ from the $G_{r_-}$ constraint. Trying $h = \alpha r_+ + \beta r_- + \gamma \lambda$ and matching coefficients of $r_+, r_-, \lambda$ on both sides forces $\alpha = \beta = -1/4$ and $\gamma = -1/2$, giving
\begin{equation}\label{eq:G-cont}
G_\lambda(r_+, r_-) = -\tfrac{1}{4}(r_+ + r_- + 2\lambda)\,\sqrt{(\lambda - r_+)(\lambda - r_-)}.
\end{equation}
In $(u, \rho)$ coordinates, $(\lambda - r_+)(\lambda - r_-) = (\lambda - u)^2 - 4\rho$ and $r_+ + r_- + 2\lambda = 2(u + \lambda)$, so
\begin{equation}\label{eq:F-G-cont-uR}
F_\lambda(u, \rho) = \sqrt{(\lambda - u)^2 - 4\rho}, \qquad G_\lambda(u, \rho) = -\tfrac{1}{2}(\lambda + u)\,\sqrt{(\lambda - u)^2 - 4\rho}.
\end{equation}

The expression $F_\lambda^2 = (\lambda - u)^2 - 4\rho$ is the discriminant of the Brunelli--Das spectral curve \eqref{eq:SW-quad}. Solving $L = \lambda$ for $p$,
\begin{equation}\label{eq:p-branches}
p_\pm(u, \rho; \lambda) = \tfrac{1}{2}\bigl[(\lambda - u) \pm \sqrt{(\lambda - u)^2 - 4\rho}\bigr] = \tfrac{1}{2}\bigl[(\lambda - u) \pm F_\lambda\bigr],
\end{equation}
so $p_+ - p_- = F_\lambda$. The algebraic continuum \eqref{eq:F-G-cont-uR} is therefore the Brunelli--Das pair recast in hodograph coordinates rather than a genuinely new Lax representation. What is gained is that $F_\lambda$ is irrational in $p$ and lives outside any Laurent layer in $p$. This is the shallow-water analogue of the observation in Section~\ref{sec:burgers} that the $S_x^m = au + \lambda$ family of inviscid Burgers is generated by fractional-power conservation laws and is invisible to a polynomial-in-$S_x$ search.

\paragraph{The dressing transform.}

The two families above are connected by a third construction. The EPD equation is linear, so any superposition of solutions is again a solution. For any contour $C \subset \C$ and any measure $m(\lambda)\, d\lambda$ for which the integrals below converge,
\begin{equation}\label{eq:dressing-F}
F_m(u, \rho) = \int_C m(\lambda)\, F_\lambda(u, \rho)\, d\lambda, \qquad G_m(u, \rho) = \int_C m(\lambda)\, G_\lambda(u, \rho)\, d\lambda
\end{equation}
solve the EPD equation and the conservation law $\partial_t F_m = \partial_x G_m$ respectively, by (formally) passing the $(x, t)$-derivatives through the $\lambda$-integral. The hodograph pair $S_x = F_m$, $S_t = G_m$ is therefore a valid Lax representation for every choice of $(C, m)$. Two specializations make this concrete.

If $m$ is a finite sum of point masses, $m(\lambda) = \sum_{j=1}^N \alpha_j\, \delta(\lambda - \lambda_j)$, then $F_m = \sum_{j=1}^N \alpha_j F_{\lambda_j}$ is a linear combination of the algebraic continuum at $N$ chosen spectral points. For $N = 2$ with $\alpha_1 = -\alpha_2 = 1$ and distinct $\lambda_1, \lambda_2$,
\[
F(u, \rho) = \sqrt{(\lambda_1 - u)^2 - 4\rho} - \sqrt{(\lambda_2 - u)^2 - 4\rho}
\]
is irrational in $(u, \rho)$ and inherits its content from two distinct points on the spectral curve. It is a perfectly valid hodograph Lax representation that no Laurent-in-$p$ ansatz could capture.

The more interesting specialization extracts Laurent coefficients of $F_\lambda$ at $\lambda = \infty$. Choose $C$ to be a counterclockwise circle of radius $R \to \infty$ enclosing the branch cut of $F_\lambda$, and take $m(\lambda) = \lambda^{-n-2}/(2\pi i)$ for $n \geq 0$. For $|\lambda| > |u| + 2\sqrt{\rho}$ the radicand is positive, and
\[
F_\lambda(u, \rho) = \lambda \sqrt{1 - \frac{2u}{\lambda} + \frac{u^2 - 4\rho}{\lambda^2}}.
\]
Applying the binomial series $\sqrt{1 + w} = 1 + w/2 - w^2/8 + w^3/16 - \cdots$ with $w = -2u/\lambda + (u^2 - 4\rho)/\lambda^2$ and collecting powers of $\lambda$ produces
\begin{equation}\label{eq:F-asymptotic}
F_\lambda(u, \rho) = \sum_{n \geq 0} \mathcal{F}_n(u, \rho)\, \lambda^{1-n}.
\end{equation}
The first three coefficients come straight from bookkeeping: at order $\lambda^1$ the leading $1$ inside the root gives $\mathcal{F}_0 = 1$; at order $\lambda^0$ only $w/2$ contributes, giving $\mathcal{F}_1 = -u$; at order $\lambda^{-1}$ the term $(u^2 - 4\rho)/(2\lambda^2)$ from $w/2$ combines with $-u^2/(2\lambda^2)$ from $-w^2/8$ to give $\mathcal{F}_2 = -2\rho$. Higher orders continue mechanically. Similarly,
\begin{equation}\label{eq:G-asymptotic}
G_\lambda(u, \rho) = -\tfrac{1}{2}(\lambda + u)\, F_\lambda(u, \rho) = \sum_{n \geq 0} \mathcal{G}_n(u, \rho)\, \lambda^{2-n}.
\end{equation}
The Cauchy residue theorem identifies these Laurent coefficients as contour integrals,
\[
\mathcal{F}_n = \frac{1}{2\pi i}\oint_{|\lambda| = R} \lambda^{n-2}\, F_\lambda\, d\lambda,
\qquad
\mathcal{G}_n = \frac{1}{2\pi i}\oint_{|\lambda| = R} \lambda^{n-3}\, G_\lambda\, d\lambda,
\]
so each $\mathcal{F}_n, \mathcal{G}_n$ is the dressing transform \eqref{eq:dressing-F} of the corresponding power-law measure. The first six pairs are
\begin{equation}\label{eq:F-G-table}
\renewcommand{\arraystretch}{1.15}
\begin{array}{c|cccccc}
n & 0 & 1 & 2 & 3 & 4 & 5 \\
\hline
\mathcal{F}_n & 1 & -u & -2\rho & -2u\rho & -2\rho(\rho + u^2) & -2u\rho(3\rho + u^2), \\
\mathcal{G}_n & -\tfrac{1}{2} & 0 & \rho + \tfrac{1}{2}u^2 & 2u\rho & \rho(\rho + 2u^2) & 2u\rho(2\rho + u^2).
\end{array}
\end{equation}

Substituting \eqref{eq:F-asymptotic} and \eqref{eq:G-asymptotic} into $\partial_t F_\lambda = \partial_x G_\lambda$ and matching coefficients of $\lambda^k$ yields a tower of conservation laws on the shallow-water flow. The factor of $\lambda$ in $G_\lambda = -\tfrac{1}{2}(\lambda + u) F_\lambda$ offsets the Laurent indexing by one: the term in $F_\lambda$ at power $\lambda^k$ is $\mathcal{F}_{1-k}$, the term in $G_\lambda$ at the same power is $\mathcal{G}_{2-k}$, and density $\mathcal{F}_n$ pairs with flux $\mathcal{G}_{n+1}$:
\begin{equation}\label{eq:Whitham-pairing}
\partial_t \mathcal{F}_n = \partial_x \mathcal{G}_{n+1}, \qquad n \geq 0.
\end{equation}
The first three nontrivial cases, read off from \eqref{eq:F-G-table}, are
\begin{align*}
n = 1: &\quad \partial_t(-u) = \partial_x\bigl(\rho + \tfrac{1}{2}u^2\bigr) &\text{(momentum)},\\
n = 2: &\quad \partial_t(-2\rho) = \partial_x(2u\rho) &\text{(mass)},\\
n = 3: &\quad \partial_t(-2u\rho) = \partial_x\bigl(\rho(\rho + 2u^2)\bigr) &\text{(energy)}.
\end{align*}
These are the classical momentum, mass, and energy conservation laws of shallow water, three consecutive rungs in an infinite tower generated by the single algebraic function $F_\lambda = \sqrt{(\lambda - u)^2 - 4\rho}$. The remainder of the tower is the Kupershmidt--Manin / Whitham hierarchy \eqref{eq:Whitham-generators}, up to overall sign and normalization inherited from the $F_\lambda$ ansatz. The hodograph Lax pair at level $n$ is $S_x = \mathcal{F}_n$, $S_t = \mathcal{G}_{n+1}$.

This section shows that each measure on $\C \cup \{\infty\}$ defines a hodograph Lax pair through \eqref{eq:dressing-F}, so shallow water has an infinite-dimensional family of Lax representations. Power-law measures at $\infty$ produce the Whitham hierarchy and the classical conservation laws via Cauchy residues. Point masses at finite $\lambda$ produce the algebraic continuum $F_\lambda$ and combinations $\sum_j \alpha_j F_{\lambda_j}$. Smooth densities on a contour interpolate between these extremes. Thus, the ``Lax pairs galore'' phenomenon of Calogero and Nucci \cite{CalogeroNucci} appears here as a function space of representations parameterized by measures on $\C \cup \{\infty\}$.

\section{A First-Order Lax Operator for the KdV equation}
\label{sec:kdv}

The four case studies so far have exhibited Lax pairs that are anomalous in various limited senses: pairs that are spectrally collapsed, infinite-parametric, or arranged into a function space of representations
parameterized by a measure.  This section presents a case where SILO returned something substantially farther from the classical picture: a \emph{first-order} scalar Lax operator for the KdV equation, living alongside the classical second-order Schr\"odinger operator in the Lax landscape of KdV. This first-order pair was first reported in~\cite{SILO}, and numerical details of that study can be found there. Here, we study some of its properties.

The first-order pair is interesting for two distinct reasons.  Diagnostically, it is by every classical measure too simple to carry integrable structure, since, as we will show, its $L^2$ point spectrum is empty, its periodic eigenvalues are uniformly spaced and encode only the mean of~$u$, and its monodromy reduces to the mass. Despite this, the full KdV conservation hierarchy is still present in the operator algebra of its powers, through a mechanism that connects the operator to the classical Bell-polynomial formalism for soliton equations \cite{LambertSpringael2008, LambertLebleSpringael2001, Fan2011} once the connection is made explicit. 
As mentioned in the introduction, the verdict that such conservation-law pairs are spectrally inert is itself longstanding, going back to the 1980s.  What the present section shows is that this inertness coexists with a full operator-algebraic recovery of the conservation hierarchy, so the standard verdict, while correct about the spectrum, is incomplete about the algebra.

\subsection{The first-order Lax pair and its fakeness}
\label{sec:fo_pair}

The KdV equation in the comoving frame reads
\begin{equation}\label{eq:kdv_std}
u_t = 6uu_x + u_{xxx}.
\end{equation}
Its classical Lax pair, due to Lax \cite{Lax1968}, is built from the self-adjoint Schr\"odinger operator $L_{\mathrm{Schr}} = -u-\partial_x^2$ and the skew-adjoint $P_{\mathrm{Schr}} = -4\partial_x^3 - 6u\partial_x - 3u_x$, with $\partial_t L = [L,P]$ equivalent to \eqref{eq:kdv_std}.  The spectral theory of~$L_{\mathrm{Schr}}$ provides the machinery for exact solution and, through the trace formula $d[\tr(L_{\mathrm{Schr}}^s)]/dt = 0$, generates the infinite hierarchy of KdV conservation laws via the Gel'fand--Dikii polynomials \cite{GelfandDikii}.

Applied to KdV with a wide operator hypothesis, SILO recovers this classical pair and, alongside it, a second pair of entirely different character, written here in an algebraically convenient form:
\begin{equation}\label{eq:fo_pair}
L = u + \partial_x, \qquad P = \kappa u + 3u^2 + u_{xx} + \kappa\,\partial_x=3u^2 + u_{xx} + \kappa L,
\end{equation}
where $\kappa\in\mathbb{R}$ is a free parameter.

\begin{theorem}[First-order Lax pair for KdV]\label{thm:fo_lax}
The pair \eqref{eq:fo_pair} satisfies $\partial_t L = [L,P]$ on solutions of \eqref{eq:kdv_std}, acting on $C^3(\R)$.
\end{theorem}

\begin{proof}
Write $L = u + D$ with $D = \partial_x$, and $P = (3u^2 + u_{xx}) + \kappa L$. Since $[L, \kappa L] = 0$,
\[
[L, P]\varphi = [L,\, 3u^2 + u_{xx}]\varphi = (6 u u_x + u_{xxx})\varphi,
\]
and equating with $(\partial_t L)\varphi = u_t \varphi$ gives \eqref{eq:kdv_std}.
\end{proof}

Two different diagnostics show that the pair \eqref{eq:fo_pair} is fake.
The Krichever integrability test~\cite{Krichever1977} compares two
dispersion laws obtained by linearizing at a constant background $u =
u_0$: the law $\omega_{\mathrm{KdV}}(\xi)$ coming from the linearized
KdV equation, and the law $\omega_{\mathrm{Lax}}(\xi)$ coming from
the auxiliary linear system $L\varphi = \lambda\varphi$, $\varphi_t =
-P\varphi$. For a Lax pair carrying algebro-geometric content, the two
agree (up to overall normalization) and are polynomial in the wavenumber
$\xi$. For our first-order pair, $\omega_{\mathrm{KdV}}(\xi) = \xi^3 -
6u_0\xi$ is cubic, while a short calculation (Appendix~\ref{app:krichever})
gives $\omega_{\mathrm{Lax}}(\xi) = -3iu_0^2$, constant in $\xi$. The
pair satisfies Krichever's polynomiality requirement only in the
degenerate sense that a constant is a polynomial of degree zero, and
encodes none of the cubic $\xi$-structure that the KdV dispersion
carries.

Additionally, the first-order pair is considered fake according to the Butler and Hay diagnostic~\cite{ButlerHay}. Setting $S = e^U$ with $U = \int u\,dx$,
the gauge transformation $\varphi \mapsto S\varphi$ conjugates the Lax
operator to
\[
S L S^{-1} = e^U (u + \partial_x)\,e^{-U} = \partial_x,
\]
a constant-coefficient operator in which every trace of $u$ has been
removed. Butler and Hay take precisely this reducibility as a sufficient
condition for fake-ness: a Lax pair gauge-equivalent to one independent
of the dependent variable cannot have been encoding the nonlinear
dynamics, however neatly the compatibility identity is satisfied. The
same verdict follows from their second test, an excess-freedom count on
the auxiliary system, which Appendix~\ref{app:butlerhay} works out
alongside the gauge construction.

\subsection{Spectral theory of the first-order operator}
\label{sec:spectral}

\paragraph{Spectrum on $\mathbb{R}$.} Before turning to the algebraic structure that carries the KdV content, we record the spectral picture for the operator~$L$ defined in \eqref{eq:fo_pair}. To this end, let 
\be
\label{e:KdV_Udef}
U(x) := \int_0^x u(s)\,ds.
\ee
The eigenvalue problem $L\varphi = \lambda\varphi$ reads $\varphi'(x) = (\lambda - u(x))\varphi(x)$, with explicit solution 
\be
\label{eq:per_sol}
\varphi(x;\lambda) = C\exp(\lambda x - U(x))\,,
\ee
with $U(x)$ given in~\eqref{e:KdV_Udef}
and $C$ an arbitrary constant.
The following proposition might seem obvious, but it is nevertheless useful to write down the result.
\begin{proposition}\label{prop:noL2}
For localized $u \in L^1(\R)$, the operator $L = u + \partial_x$ has no $L^2(\R)$ eigenvalues for any $\lambda \in \C$.
\end{proposition}

\begin{proof}
As $|x| \to \infty$, $U(x) \to U_\pm$, which is finite for $u\in L^1(\mathbb{R}),$ so $|\varphi(x;\lambda)| \sim |C_\pm|\,e^{(\operatorname{Re}\lambda)x}$.  Membership in $L^2(\R)$ would require $\operatorname{Re}\lambda < 0$ (to decay at $+\infty$) and $\operatorname{Re}\lambda > 0$ (to decay at $-\infty$), which is impossible to combine.  Hence no~$\lambda \in \C$ admits an $L^2$ eigenfunction.
\end{proof}

Another direct computation reinforces this degenerate spectral picture. The eigenfunction \eqref{eq:per_sol} has the normalized transfer multiplier
\[
\lim_{x \to +\infty} e^{-\lambda x}\varphi(x;\lambda)\,\Big/\,\lim_{x \to -\infty} e^{-\lambda x}\varphi(x;\lambda)
=
\exp\!\left(-\!\int_{-\infty}^{\infty}u\,dx\right),
\]
which is independent of $\lambda$ and encodes only the mass. Thus even the natural first-order analogue of scattering data contains no spectral dependence. In particular, there is no counterpart of the nontrivial reflection/transmission data that drive the Schr\"odinger inverse-scattering theory \cite{AblowitzSegur,FaddeevTakhtajan}.

\paragraph{Periodic spectrum.} 
 Without loss of generality (thanks to the scaling invariance of the KdV equation), we choose the spatial period to be $2\pi$. 
The periodic spectrum of $L$ is characterized by the following:
\begin{theorem}\label{thm:per_spec}
On the interval $[0,2\pi]$ with periodic boundary conditions, the eigenvalues of~$L$ are $\lambda_n = \bar u + i n$ for $n \in \Z$, with eigenfunctions $\varphi_n(x) = w(x)\,e^{i n x}$ where $w(x) = \exp(-\int_0^x(u - \bar u)\,ds)$.
\end{theorem}

\begin{proof}
Recall that the general solution of the eigenvalue problem $L\varphi = \lambda\varphi$  is given by~\eqref{eq:per_sol} with
$C := \varphi(0)$.
To impose periodicity, separate the mean of $u$ from its fluctuation. Set
\[
\bar u := \frac{1}{2\pi}\int_0^{2\pi} u(s)\,ds, \qquad U_0(x) := \int_0^x \bigl(u(s) - \bar u\bigr)\,ds.
\]
Then $\int_0^x u(s)\,ds = \bar u\, x + U_0(x)$, and \eqref{eq:per_sol} factors as
\begin{equation}\label{eq:per_factor}
\varphi(x;\lambda) = C \exp\!\bigl((\lambda - \bar u) x\bigr)\, w(x), \qquad w(x) := e^{-U_0(x)}.
\end{equation}
The function $w$ is $2\pi$-periodic, since
\[
U_0(x + 2\pi) - U_0(x) = \int_x^{x+2\pi}\bigl(u(s) - \bar u\bigr)\,ds = \int_0^{2\pi}\bigl(u(s) - \bar u\bigr)\,ds = 0
\]
by translation invariance of the integral over a full period and the definition of $\bar u$. In particular, $w(0) = w(2\pi) = 1$.

The periodic boundary condition $\varphi(2\pi;\lambda) = \varphi(0;\lambda)$ applied to \eqref{eq:per_factor} yields
\(
C \exp\!\bigl(2\pi(\lambda - \bar u)\bigr) \cdot w(2\pi) = C \cdot w(0)
\),
which, after cancelling $C$ and using $w(2\pi) = w(0) = 1$, reduces to
\(
\exp\!\bigl(2\pi(\lambda - \bar u)\bigr) = 1
\).
The complex exponential equals $1$ precisely when its argument is an integer multiple of $2\pi i$, so
\(2\pi(\lambda - \bar u) = 2\pi i\, n\) for $n \in \Z$,
giving $\lambda_n = \bar u + i n$. Substituting $\lambda = \lambda_n$ into \eqref{eq:per_factor} and absorbing $C$ into the normalization,
\(
\varphi_n(x) = w(x)\, e^{i n x},
\)
which is the claimed eigenfunction.
\end{proof}

Thus, the eigenvalues $\lambda_n = \bar u + i n$ lie on a single vertical line in
$\C$ at uniform spacing, with real part fixed by the mean $\bar u$ and
imaginary part labelled by $n \in \Z$. Consequently, the eigenfunctions $\varphi_n(x) = w(x)\, e^{i n x}$
are just weighted Fourier modes
with weight $w(x) = e^{-U_0(x)}$. There are no bands, no gaps, no nontrivial
Floquet discriminant~\cite{MagnusWinkler}.  The Schr\"odinger-Gel'fand-Dikii route~\cite{GelfandDikii}, which extracts
the KdV hierarchy from the resolvent expansion of $L_{\mathrm{Schr}}$, has no analogue here.


\subsection{From the first-order pair to the KdV conservation hierarchy}
\label{sec:operator_algebra}

The diagnostics in Section~\ref{sec:fo_pair} declare the first-order pair
\eqref{eq:fo_pair} fake, yet the KdV conservation hierarchy is encoded in the
operator algebra generated by powers of $L$ in the ungauged frame. We set
$\kappa = 0$ in \eqref{eq:fo_pair} to drop the $\kappa L$ tail, which lies in
the centralizer of $L$ and contributes nothing, leaving
\begin{equation}\label{eq:fo_pair_clean}
(L, P) = \bigl(u + \partial_x,\ 3u^2 + u_{xx}\bigr).
\end{equation}
Here $P = 3u^2 + u_{xx}$ is the KdV flux, and the Lax equation collapses to the
conservation form of \eqref{eq:kdv_std},
\[
[L, P] = [\partial_x,\, 3u^2 + u_{xx}] = 6 u u_x + u_{xxx} = u_t,
\qquad u_t = \partial_x P.
\]

\begin{remark}
This is a special case of a general freedom in the Lax representation, noted to us by B.~Konopelchenko~\cite{KonopelchenkoPC}. The same equation is obtained from a modified pair in which $P\mapsto P + F(L)$ for any function $F$, since $[L, F(L)] = 0$, so the time operator is determined only modulo the centralizer of $L$. The $\kappa L$ tail is the linear instance $F(L) = \kappa L$. This centralizer freedom is the operator-algebra counterpart of the recurring lesson of this paper, that compatibility underdetermines the Lax representation, joining the gauge, spectral-shift, and hodograph freedoms encountered in the previous sections.
\end{remark}

The machinery that extracts the hierarchy from this pair is classical, and we
use it as such. The densities produced by $L$ are Bell polynomials
\cite{Bell1934, Comtet1974}, conserved densities are read off by the direct
method of G\"okta\c{s} and Hereman \cite{GoktasHereman1997}, which rests on the
variational fact that a density is conserved exactly when its time derivative is
a total spatial derivative \cite{Olver1993}, and the Bell-polynomial form of
this calculation is standard for soliton equations
\cite{LambertSpringael2008, GilsonLambertNimmoWillox1996, Fan2011}. What we add
is the operator that feeds the machine. The densities are manufactured by the
spectrally fake pair \eqref{eq:fo_pair_clean}, and the construction asks nothing
of the equation beyond its conservation form $u_t = \partial_x P$, so the same
device runs on any scalar evolution equation written that way. We keep the
exposition self-contained so the procedure can be run without chasing
references, but none of the underlying identities below are new.

\paragraph{Bell polynomials.}
The operator $L = u + \partial_x$ has the gauge presentation
\begin{equation}\label{eq:gauge_L}
L = e^{-U}\,\partial_x\,e^{U}, \qquad U := \int u\, dx,
\end{equation}
checked directly: $e^{-U}\partial_x(e^U f) = u f + f_x = L f$. Induction gives
$L^s = e^{-U}\partial_x^s e^U$, and acting on the constant function $1$,
\begin{equation}\label{eq:rho_gauge}
\rho_s := L^s \cdot 1 = e^{-U}\,\partial_x^s\bigl(e^U\bigr)
        = Y_s\bigl(u, u_x, \ldots, u^{(s-1)}\bigr),
\end{equation}
the $s$-th complete exponential Bell polynomial \cite{Bell1934, Comtet1974},
with first values
\begin{align*}
\rho_1 &= u, \qquad \rho_2 = u^2 + u_x, \qquad \rho_3 = u^3 + 3 u u_x + u_{xx},\\
\rho_4 &= u^4 + 6u^2 u_x + 3 u_x^2 + 4 u u_{xx} + u_{xxx},\\
\rho_5 &= u^5 + 10 u^3 u_x + 15 u u_x^2 + 10 u^2 u_{xx} + 10 u_x u_{xx} + 5 u u_{xxx} + u_{xxxx}.
\end{align*}
The complete Bell polynomials and their cumulant inverse are mutually inverse
polynomial changes of variable \cite{Comtet1974}, so
\begin{equation}\label{eq:inverse_bell}
u^{(j)} = \widehat{Y}_{j+1}(\rho_1, \ldots, \rho_{j+1}), \qquad j \geq 0,
\end{equation}
with first instances $u = \rho_1$, $u_x = \rho_2 - \rho_1^2$, and
$u_{xx} = \rho_3 - 3\rho_1 \rho_2 + 2\rho_1^3$. The generating sets
$\{u, u_x, u_{xx}, \ldots\}$ and $\{\rho_1, \rho_2, \ldots\}$ describe the same
polynomials in $u$ and its derivatives, and we move freely between them.

\paragraph{The two recursions.}
Both $\partial_t$ and $\partial_x$ act on the $\rho_s$ by explicit polynomial
rules. For the time derivative, the Lax equation $\partial_t L = [L, P]$
telescopes to $\partial_t L^s = [L^s, P]$. Applying this to $1$ gives
$\partial_t \rho_s = L^s P - P \rho_s$, and expanding
$L^s P = e^{-U}\partial_x^s(e^U P)$ by Leibniz with
$\partial_x^k e^U = \rho_k e^U$ leaves, after the $j = 0$ term cancels
$P\rho_s$,
\begin{equation}\label{eq:dt_rho}
\boxed{\;\partial_t \rho_s
   = \sum_{j=1}^{s} \binom{s}{j}\, \rho_{s-j}\, P^{(j)},
   \qquad P^{(j)} := D_x^j P, \quad s \geq 1.\;}
\end{equation}
The flux derivatives are explicit, $P^{(0)} = 3u^2 + u_{xx}$,
$P^{(1)} = 6 u u_x + u_{xxx}$, $P^{(2)} = 6 u_x^2 + 6 u u_{xx} + u_{xxxx}$, and
so on, and with Leibniz \eqref{eq:dt_rho} propagates $\partial_t$ to every
polynomial in the $\rho_s$. For the space derivative, $L\rho_s = \rho_{s+1}$ and
$L = u + \partial_x$ give
\begin{equation}\label{eq:Drho}
(\rho_s)_x = \rho_{s+1} - \rho_1\, \rho_s,
\end{equation}
which evaluates $D_x$ inside $R_\rho := \Q[\rho_1, \rho_2, \ldots]$. This
$D_x$ is the ordinary total $x$-derivative written in the $\rho$
variables, as one sees by differentiating the generating function
$A(t) = e^{U(x+t) - U(x)}$ in $x$, so reducing modulo total $x$-derivatives
means the same thing in either generating set.

\paragraph{Reading off the hierarchy.}
A local functional $\int F\, dx$ is conserved exactly when $\partial_t F$ is a
total $x$-derivative \cite{Olver1993}, so the conserved densities are the kernel
of $\partial_t$ acting on $R_\rho$ modulo $D_x R_\rho$, and this kernel is
a finite computation level by level. Grade $R_\rho$ by subscript sum, with
$\rho_s$ in degree $s$ and $\rho_{i_1}\cdots\rho_{i_k}$ in degree
$i_1 + \cdots + i_k$; the monomials of subscript sum at most $N$ are finite in
number. Take the ansatz
\[
F = \sum_{|\alpha| \leq N} c_\alpha\, m_\alpha, \qquad
m_\alpha = \rho_{i_1}\cdots\rho_{i_k}, \quad |\alpha| = i_1 + \cdots + i_k,
\]
compute $\partial_t F$ from \eqref{eq:dt_rho} with Leibniz, and reduce modulo
$D_x R_\rho$ by \eqref{eq:Drho}, which is integration by parts. Requiring
the surviving residue monomials to cancel, $\partial_t F \equiv 0$, is a
homogeneous linear system $M c = 0$ over $\Q$, and $\ker M$ is the space of KdV
conservation densities of subscript sum at most $N$. This is the
G\"okta\c{s}--Hereman direct method \cite{GoktasHereman1997} carried out in the
Bell variables natural to $L$, in exact rational arithmetic, with no spectral
information about $L$ entering at any step. Appendix~\ref{app:KdVCQ} states the
procedure as an algorithm and runs it to recover the first five KdV conserved
quantities.

The hierarchy that comes out is the standard one, and the route to it through
the Schr\"odinger spectrum is unnecessary, as Miura, Gardner, and Kruskal showed
by other means \cite{MiuraGardnerKruskal1968}. The point we draw is the
provenance. The densities are manufactured by the pair \eqref{eq:fo_pair_clean},
which Section~\ref{sec:fo_pair} certifies as fake by the Krichever and
Butler--Hay tests, and the conservation content rides on the operator algebra of
$L$ in a fixed gauge, the structure those gauge-invariant tests are built to
discard. A pair with no spectral content still delivers the full hierarchy.

\section{Concluding Remarks}

The five case studies presented in this work  share one common  lesson: compatibility under-determines the Lax representation. 
A summary of the five case studies pursued in this paper and the SILO variant under which each pair was surfaced is given in Table~\ref{tab:silo-variants}.
Every equation studied in this work is classically integrable or linearizable. Thus, the object of study here was not integrability,  but rather the shape of the representations that compatibility allows, and 
the present study demonstrated that many possible scenarios can arise. 

The Euler top returned rank-deficient pairs that satisfy $\dot L = [L,P]$ while encoding
less than the dynamics, with skew-symmetry and a Frobenius product penalty needed
to surface a Manakov skeleton whose spectral completion is then linear algebra on
$\mathfrak{so}(3)$. The free Schr\"odinger equation showed the nondegeneracy
penalty itself choosing the pair: a determinant barrier returns the textbook UTM
pair, a Frobenius surrogate returns a spectrally empty pair carrying a free
function $M(k)$. Inviscid Burgers carried an infinite gauge freedom in its
Clairaut compatibility and an exponent continuum $S_x^m = au + \lambda$, both
instances of the Calogero--Nucci hodograph mechanism on the fractional-power
conservation laws. Shallow water widened that continuum into a function space of
hodograph pairs parametrized by measures on $\C \cup \{\infty\}$, with the
Brunelli--Das pair the unique Laurent representative and the Whitham hierarchy and
spectral-discriminant continuum lying outside the Laurent layer entirely. None of
these degeneracies is visible to the compatibility condition, which is satisfied
by all of them at once.

The first-order KdV pair $L = u + \partial_x$ is the sharpest case, and it is
worth ending on. By every gauge-invariant diagnostic it is fake. The Krichever dispersion
collapses to a constant (Appendix~\ref{app:krichever}), and both Butler--Hay tests apply (Appendix~\ref{app:butlerhay}): the gauge $S = e^U$ sends $L$ to the bare $\partial_x$, and the excess-freedom count leaves the flux $q$
unconstrained. After that gauge $\partial_x^s \cdot 1 = 0$ and the whole construction collapses,
exactly as the classification predicts. Yet in the fixed gauge $v = \int u$ the
same pair drives the classical direct method to the full KdV hierarchy, $I_1$
through $I_5$ at subscript-sum levels $1, 2, 4, 6, 8$ and onward. The two verdicts
do not contradict each other. They measure different things. The gauge-invariant
tests read the spectral content, which is empty, while the conservation content
lives in the operator algebra of $L$ in a fixed gauge, the structure those tests
are built to quotient away.

We do not propose a refined taxonomy. The Butler--Hay test is correct for what it
tests, the conservation content here is gauge-dependent, and the two simply do not
see each other. Butler and Hay anticipated the gap, writing that ``the extent to
which a fake Lax pair can be used to gain information about its associated system
is also open to debate.'' Section~\ref{sec:kdv} is one answer for the
conservation-form class: a pair declared fake on gauge-invariant grounds still
carries the full conservation ring of an arbitrary $u_t = q_x$ when read in the
gauge those grounds discard. Whether other fake pairs in the literature hide
algebraic content of this kind, and whether the Butler--Hay criteria admit a
gauge-fixed refinement separating the algebraically active fakes from the empty
ones, are questions we leave open.

The machinery doing the work is classical, as Section~\ref{sec:operator_algebra} points out. The hierarchy is manufactured by a spectrally fake pair, and the
construction needs nothing of the equation past its conservation form, so the same
device runs on any scalar $u_t = q_x$. The universality is a special case of
Calogero--Nucci, and the combination we use, the substrate $(u + \partial_x,\, q)$
in the gauge $v = \int u$, is elementary enough that we make no claim of priority.

\begin{table}[t!]
\centering
\small
\renewcommand{\arraystretch}{1.45}
\setlength{\tabcolsep}{6pt}
\begin{tabular}{|
  >{\centering\arraybackslash}m{3.2cm}|
  >{\centering\arraybackslash}m{3.2cm}|
  >{\centering\arraybackslash}m{3.8cm}|
  >{\centering\arraybackslash}m{4.4cm}|}
\hline
\textbf{System} & \textbf{Ansatz} & \textbf{Compatibility} & \textbf{Degeneracy control} \\
\hline
$I\dot{\Omega}+\Omega\times(I\Omega)=0$
  & Linear in $\Omega$, skew-symmetric
  & Lie--Poisson chain rule $\{H,L\}=[L,P]$
  & Frobenius product penalty $\prod_k \|\partial_{\Omega_k}\tilde L\|_F^2\,\|\partial_{\Omega_k}\tilde P\|_F^2$ \\
\hline
$iq_t+q_{xx}=0$
  & Affine $(\tilde f,\tilde g)$, polynomial in $k$
  & Clairaut $\mu_{xt}=\tilde f_t=\tilde g_x=\mu_{tx}$
  & $|\det(Z^\dagger Z)|^{-2}$ barrier, or Frobenius penalty $\|Z^+\|_F^2$ \\
\hline
$u_t+u u_x=0$
  & Polynomial in $(u,s)$, $s=S_x$
  & Clairaut $2s\,g_u+f_u(g_s+u)=0$
  & $\ell^\infty$ product penalty $\|\Xi\|_\infty^{-1}\|\mathcal{Z}\|_\infty^{-1}$ \\
\hline
$\begin{aligned}\rho_t+(\rho u)_x &= 0\\ u_t+u u_x+\rho_x &= 0\end{aligned}$
  & $L$ Laurent in $p$, $B=(\tfrac{1}{2}L^2)_{>0}$
  & Chain rule on $(u,\rho)$, $L_t=\{L,B\}$
  & Inverse-mean product $1/\mathbb{E}_\nu[\widetilde L_u^2\,\widetilde L_\rho^2]$ \\
\hline
$u_t-6u u_x-u_{xxx}=0$
  & $\tilde L=\xi_1 u+\xi_2\partial_x+\xi_3\partial_x^2$, $\tilde P$ general tensor
  & Hamiltonian chain rule $\partial_t\tilde L=(\partial_u\tilde L)\,Q\,\tfrac{\delta H}{\delta u}=[\tilde L,\tilde P]$
  & Inverse-norm penalty $\|(\partial_u\tilde L)\,Q(\delta H/\delta u)\|_2^{-2}$ \\
\hline
\end{tabular}
\caption{A summary of the five case studies pursued throughout this paper (top to bottom: Euler top, free Schr\"odinger, inviscid Burgers, shallow water, first-order KdV) and the SILO variant under which each pair was surfaced. Degeneracy control takes a different form in every case, and each variant reflects the structural geometry of the underlying problem. The first-order KdV formulation and numerical details are from \cite{SILO} and not rerun here.}
\label{tab:silo-variants}
\end{table}

The broader point returns to where the paper started. Anomalous Lax pairs are not
pathologies at the edge of the theory. They are regular inhabitants of the Lax
landscape that any discovery procedure, sparse regression or otherwise, will keep
returning, because compatibility alone cannot tell a spectrally informative pair
from a degenerate or fake one. The burden of telling them apart falls on the
search, on the ansatz, the gauge, and the penalty geometry, not on the
compatibility condition, which certifies all of them equally.
As data-driven methods providing such certification continue to develop,
it will be important to continue expanding the relevant detailed pair
characterization and whether the latter can be used to extract the 
infinite hierarchy of conservation laws (of the integrable system) or/and
whether it can be leveraged to develop systematically the inverse scattering
analysis of the system at hand.
A complementary question is whether these fake pairs, beyond carrying the conservation hierarchy, can be used to constructively recover the genuine spectral pair of the underlying equation. For scalar conservation laws of Gelfand--Dikii type, this turns out to be possible. The first-order operator supplies the differential skeleton through its antisymmetric part $\partial_x = \tfrac12(L - L^\dagger)$, its conserved-quantity gradients supply the field content, and the Krichever test certifies the result. We will develop this constructive converse, which is the closest scalar analogue of the Manakov completion of Section~\ref{sec:euler}, in forthcoming work.

\section*{Acknowledgements}

J.A. gratefully acknowledges helpful discussions during the early stages of this work with Roy H. Goodman and Luis García-Naranjo about integrable rigid body dynamics and Manakov completions. J.A. also thanks Ildar Gabitov for pointing out the Krichever test and Oleksandr Minakov and Ibrahim Fatkullin for separate discussions concerning the spectral theory of the first order KdV Lax pair.
The authors thank Boris Konopelchenko for a careful reading of an earlier draft and for the historical and structural comments that sharpened Sections~\ref{sec:kdv} and the discussion of representation freedom.
 J.A. acknowledges support from NSF grant PHY-2316622.
W.Z.\ acknowledges support from the National Science Foundation under awards DMS-2502900 and DMS-2540370, and from the Air Force Office of Scientific Research under Grant No. FA9550-25-1-0079.
G.B.\ acknowledges support from the Simons foundation under grant number SFI-MPS-TSM-00013369.
This material is based upon work supported by the National Science Foundation under Grant No.
PHY-2408988 (PGK). 
This research was partly conducted while P.G.K. was  visiting the Okinawa Institute of Science and
Technology (OIST) through the Theoretical Sciences Visiting Program (TSVP), the University of
Sydney through the visitor program of the Sydney Mathematical Research Institute (SMRI) and the Department of Mechanical Engineering at Seoul National University through a Fulbright Fellowship. Their support is gratefully acknowledged.
This work was also  supported by a grant from the Simons Foundation [SFI-MPS-SFM-00011048, P.G.K]. 
The authors collectively thank the Banff International Research Station (BIRS) at the Instituto de Matemáticas de la Universidad de Granada and the Institute for Computational and Experimental Mathematics (ICERM) where portions of this work were carried out during week long research workshops.

\bibliographystyle{abbrv}
\bibliography{refs}

@article{Lax1968,
  author  = {Lax, Peter D.},
  title   = {Integrals of nonlinear equations of evolution and solitary waves},
  journal = {Communications on Pure and Applied Mathematics},
  volume  = {21},
  number  = {5},
  pages   = {467--490},
  year    = {1968}
}

@article{Gardner1967,
  author  = {Gardner, Clifford S. and Greene, John M. and Kruskal, Martin D. and Miura, Robert M.},
  title   = {Method for solving the {K}orteweg--de {V}ries equation},
  journal = {Physical Review Letters},
  volume  = {19},
  pages   = {1095--1097},
  year    = {1967}
}

@book{AblowitzSegur,
  author    = {Ablowitz, Mark J. and Segur, Harvey},
  title     = {Solitons and the Inverse Scattering Transform},
  publisher = {SIAM},
  address   = {Philadelphia},
  year      = {1981}
}

@article{ZakharovShabat,
  author  = {Zakharov, V. E. and Shabat, A. B.},
  title   = {Exact theory of two-dimensional self-focusing and one-dimensional self-modulation of waves in nonlinear media},
  journal = {Soviet Physics JETP},
  volume  = {34},
  number  = {1},
  pages   = {62--69},
  year    = {1972}
}

@article{CalogeroNucci,
  author  = {Calogero, F. and Nucci, M. C.},
  title   = {{L}ax pairs galore},
  journal = {Journal of Mathematical Physics},
  volume  = {32},
  pages   = {72--74},
  year    = {1991}
}

@unpublished{ButlerHay,
  author        = {Butler, Samuel and Hay, Mike},
  title         = {Simple identification of fake {L}ax pairs},
  note          = {arXiv preprint},
  eprint        = {1311.2406},
  archivePrefix = {arXiv},
  primaryClass  = {nlin.SI},
  year          = {2013}
}

@article{Gubbiotti2016,
  author  = {Gubbiotti, G. and Scimiterna, C. and Levi, D.},
  title   = {Linearizability and a fake Lax pair for a nonlinear nonautonomous quad-graph equation consistent around the cube},
  journal = {Theoretical and Mathematical Physics},
  volume  = {189},
  pages   = {1459--1471},
  year    = {2016},
  doi     = {10.1134/S0040577916100068},
  url     = {https://doi.org/10.1134/S0040577916100068}
}

@Article{axioms13020121,
AUTHOR = {Burde, Georgy I.},
TITLE = {Lax Pairs for the Modified KdV Equation},
JOURNAL = {Axioms},
VOLUME = {13},
YEAR = {2024},
NUMBER = {2},
ARTICLE-NUMBER = {121},
URL = {https://www.mdpi.com/2075-1680/13/2/121},
ISSN = {2075-1680},
DOI = {10.3390/axioms13020121}
}

@article{SILO,
  author  = {Adriazola, Jimmie and Zhu, Wei and Kevrekidis, Panayotis G. and Aceves, Alejandro},
  title   = {Computer assisted discovery of integrability via {SILO}: sparse identification of {L}ax operators},
  journal = {SIAM Journal on Applied Dynamical Systems},
  volume  = {25},
  number  = {1},
  pages   = {131--159},
  year    = {2026}
}

@article{Manakov1976,
  author  = {Manakov, S. V.},
  title   = {Note on the integration of {E}uler's equations of the dynamics of an $n$-dimensional rigid body},
  journal = {Functional Analysis and Its Applications},
  volume  = {10},
  number  = {4},
  pages   = {328--329},
  year    = {1976}
}

@book{bolsinovfomenko2004,
  author    = {Bolsinov, Alexey V. and Fomenko, Anatoly T.},
  title     = {Integrable {H}amiltonian Systems: Geometry, Topology, Classification},
  publisher = {Chapman \& Hall/CRC},
  address   = {Boca Raton},
  year      = {2004}
}

@book{kozlov1993,
  author    = {Kozlov, Valery V.},
  title     = {Symmetries, Topology and Resonances in {H}amiltonian Mechanics},
  series    = {Ergebnisse der Mathematik und ihrer Grenzgebiete},
  publisher = {Springer},
  address   = {Berlin},
  year      = {1993}
}

@article{reyman1980,
  author  = {Reyman, A. G.},
  title   = {Group-theoretical methods in the theory of finite-dimensional integrable systems},
  journal = {Russian Mathematical Surveys},
  volume  = {35},
  number  = {1},
  pages   = {57--91},
  year    = {1980}
}

@article{Fokas1997,
  author  = {Fokas, A. S.},
  title   = {A unified transform method for solving linear and certain nonlinear {PDE}s},
  journal = {Proceedings of the Royal Society A},
  volume  = {453},
  number  = {1962},
  pages   = {1411--1443},
  year    = {1997}
}

@book{Fokas2008,
  author    = {Fokas, A. S.},
  title     = {A Unified Approach to Boundary Value Problems},
  publisher = {SIAM},
  year      = {2008}
}

@article{GelfandDikii,
  author  = {Gel'fand, I. M. and Dikii, L. A.},
  title   = {Asymptotic behaviour of the resolvent of {S}turm--{L}iouville equations and the algebra of the {K}orteweg--de {V}ries equations},
  journal = {Russian Mathematical Surveys},
  volume  = {30},
  pages   = {77--113},
  year    = {1975}
}

@article{Bell1934,
  author  = {Bell, E. T.},
  title   = {Exponential polynomials},
  journal = {Annals of Mathematics},
  volume  = {35},
  number  = {2},
  pages   = {258--277},
  year    = {1934},
  doi     = {10.2307/1968431}
}

@book{Comtet1974,
  author    = {Comtet, Louis},
  title     = {Advanced Combinatorics: The Art of Finite and Infinite Expansions},
  publisher = {D.~Reidel},
  address   = {Dordrecht},
  year      = {1974}
}

@article{Miura1968,
  author  = {Miura, Robert M.},
  title   = {{K}orteweg--de {V}ries equation and generalizations. {I}.
             {A} remarkable explicit nonlinear transformation},
  journal = {Journal of Mathematical Physics},
  volume  = {9},
  number  = {8},
  pages   = {1202--1204},
  year    = {1968},
  doi     = {10.1063/1.1664700}
}

@article{MiuraGardnerKruskal1968,
  author  = {Miura, Robert M. and Gardner, Clifford S. and Kruskal, Martin D.},
  title   = {{K}orteweg--de {V}ries equation and generalizations. {II}.
             {E}xistence of conservation laws and constants of motion},
  journal = {Journal of Mathematical Physics},
  volume  = {9},
  number  = {8},
  pages   = {1204--1209},
  year    = {1968},
  doi     = {10.1063/1.1664701}
}

@article{GilsonLambertNimmoWillox1996,
  author  = {Gilson, C. and Lambert, F. and Nimmo, J. J. C. and Willox, R.},
  title   = {On the combinatorics of the {H}irota {D}-operators},
  journal = {Proceedings of the Royal Society of London A},
  volume  = {452},
  number  = {1945},
  pages   = {223--234},
  year    = {1996},
  doi     = {10.1098/rspa.1996.0013}
}

@article{LambertLebleSpringael2001,
  author  = {Lambert, F. and Lebl\'{e}, S. and Springael, J.},
  title   = {Binary {B}ell polynomials and {D}arboux covariant {L}ax pairs},
  journal = {Glasgow Mathematical Journal},
  volume  = {43A},
  pages   = {53--63},
  year    = {2001},
  doi     = {10.1017/S0017089501000088}
}

@article{LambertSpringael2008,
  author  = {Lambert, Franklin and Springael, Johan},
  title   = {Soliton equations and simple combinatorics},
  journal = {Acta Applicandae Mathematicae},
  volume  = {102},
  number  = {2--3},
  pages   = {147--178},
  year    = {2008},
  doi     = {10.1007/s10440-008-9209-3}
}

@article{Fan2011,
  author  = {Fan, Engui},
  title   = {The integrability of nonisospectral and variable-coefficient {KdV}
             equation with binary {B}ell polynomials},
  journal = {Physics Letters A},
  volume  = {375},
  number  = {3},
  pages   = {493--497},
  year    = {2011},
  doi     = {10.1016/j.physleta.2010.11.038}
}

@article{WahlquistEstabrook1975,
  author  = {Wahlquist, H. D. and Estabrook, F. B.},
  title   = {Prolongation structures of nonlinear evolution equations},
  journal = {Journal of Mathematical Physics},
  volume  = {16},
  number  = {1},
  pages   = {1--7},
  year    = {1975},
  doi     = {10.1063/1.522396}
}

@book{Dickey2003,
  author    = {Dickey, L. A.},
  title     = {Soliton Equations and Hamiltonian Systems},
  edition   = {2nd},
  series    = {Advanced Series in Mathematical Physics},
  volume    = {26},
  publisher = {World Scientific},
  address   = {Singapore},
  year      = {2003},
  doi       = {10.1142/5108}
}

@article{Krichever1977,
  author  = {I. M. Krichever},
  title   = {Methods of algebraic geometry in the theory of non-linear equations},
  journal = {Russian Mathematical Surveys},
  volume  = {32},
  number  = {6},
  pages   = {185--213},
  year    = {1977},
}

@article{MishchenkoFomenko1978,
  author  = {Mishchenko, A. S. and Fomenko, A. T.},
  title   = {Euler equations on finite-dimensional {L}ie groups},
  journal = {Mathematics of the USSR-Izvestiya},
  volume  = {12},
  number  = {2},
  pages   = {371--389},
  year    = {1978},
  doi     = {10.1070/IM1978v012n02ABEH001859}
}

@book{NMPZ1984,
  author= {Novikov, S. P. and Manakov, S. V. and Pitaevskii, L. P. and Zakharov, V. E.},
  title= {Theory of solitons: the inverse scattering transform},
  publisher= {Plenum},
  year= 1984
}

@inproceedings{Zakharov1994,
  author = {V.E. Zakharov},
  title = {Dispersionless limit of integrable systems in 2 + 1 dimensions}, 
  booktitle = {Singular Limits of Dispersive Waves}, 
  editor = {Ercolani, N. M. and Gabitov, I. R. and Levermore, C. D. and Serre, D.}, 
  publisher = {Plenum Press},
  year = 1994,
  pages = {165-174}
}

@book{FaddeevTakhtajan,
  author    = {Faddeev, L. D. and Takhtajan, L. A.},
  title     = {Hamiltonian Methods in the Theory of Solitons},
  series    = {Classics in Mathematics},
  publisher = {Springer},
  address   = {Berlin},
  year      = {1987},
  doi       = {10.1007/978-3-540-69969-9}
}

@book{MagnusWinkler,
  author    = {Magnus, Wilhelm and Winkler, Stanley},
  title     = {Hill's Equation},
  publisher = {Interscience Publishers},
  address   = {New York},
  year      = {1966},
  note      = {Reprinted by Dover Publications, 1979}
}

@article{BrunelliDas1997,
  author  = {Brunelli, J. C. and Das, A.},
  title   = {A {Lax} description for polytropic gas dynamics},
  journal = {Physics Letters A},
  volume  = {235},
  number  = {6},
  pages   = {597--602},
  year    = {1997},
  doi     = {10.1016/S0375-9601(97)00708-1},
  note    = {Preprint \texttt{solv-int/9706005}}
}

@article{Brunelli2000,
  author  = {Brunelli, J. C.},
  title   = {Dispersionless limit of integrable models},
  journal = {Brazilian Journal of Physics},
  volume  = {30},
  number  = {2},
  pages   = {455--468},
  year    = {2000},
  note    = {Preprint \texttt{nlin/0207042}}
}

@article{ConstandacheDasToppan2002,
  author  = {Constandache, A. and Das, A. and Toppan, F.},
  title   = {{Lucas} polynomials and a standard {Lax} representation for the polytropic gas dynamics},
  journal = {Letters in Mathematical Physics},
  volume  = {60},
  pages   = {197--209},
  year    = {2002}
}

@article{KonopelchenkoMartinezMedina2010,
  author  = {Konopelchenko, B. and Mart{\'\i}nez Alonso, L. and Medina, E.},
  title   = {Hodograph solutions of the dispersionless coupled {KdV} hierarchies, critical points and the {Euler--Poisson--Darboux} equation},
  journal = {Journal of Physics A: Mathematical and Theoretical},
  volume  = {43},
  number  = {43},
  pages   = {434020},
  year    = {2010},
  note    = {Preprint \texttt{arXiv:1003.2892}}
}

@book{Whitham1974,
  author    = {Whitham, G. B.},
  title     = {Linear and Nonlinear Waves},
  publisher = {Wiley-Interscience},
  address   = {New York},
  year      = {1974}
}

@book{CourantHilbert,
  author    = {Courant, R. and Hilbert, D.},
  title     = {Methods of Mathematical Physics, Vol. {II}: Partial Differential Equations},
  publisher = {Wiley-Interscience},
  address   = {New York},
  year      = {1962}
}

@article{CarrierGreenspan1958,
  author  = {Carrier, G. F. and Greenspan, H. P.},
  title   = {Water waves of finite amplitude on a sloping beach},
  journal = {Journal of Fluid Mechanics},
  volume  = {4},
  number  = {1},
  pages   = {97--109},
  year    = {1958},
  doi     = {10.1017/S0022112058000331}
}

@article{Tsarev1991,
  author  = {Ts{\"a}rev, S. P.},
  title   = {The geometry of {Hamiltonian} systems of hydrodynamic type. {The} generalized hodograph method},
  journal = {Mathematics of the {USSR}-{Izvestiya}},
  volume  = {37},
  number  = {2},
  pages   = {397--419},
  year    = {1991},
  doi     = {10.1070/IM1991v037n02ABEH002069}
}

@article{KupershmidtManin1977,
  author  = {Kupershmidt, B. A. and Manin, {Yu}. I.},
  title   = {Long-wave equations with a free surface. {I}. {Conservation} laws and solutions},
  journal = {Functional Analysis and Its Applications},
  volume  = {11},
  number  = {3},
  pages   = {188--197},
  year    = {1977},
  note    = {English translation of Funkts. Anal. Prilozhen. \textbf{11} (1977), no.~3, 31--42}
}

@unpublished{Sakovich2001,
  author        = {Sakovich, S. Yu.},
  title         = {True and fake {L}ax pairs: how to distinguish them},
  note          = {arXiv preprint},
  eprint        = {nlin/0112027},
  archivePrefix = {arXiv},
  year          = {2001}
}

@book{Olver1993,
  author    = {Olver, Peter J.},
  title     = {Applications of {L}ie Groups to Differential Equations},
  series    = {Graduate Texts in Mathematics},
  volume    = {107},
  edition   = {2},
  publisher = {Springer},
  address   = {New York},
  year      = {1993}
}

@article{KrippendorfLustSyvaeri2021,
  author  = {Krippendorf, Sven and L{\"u}st, Dieter and Syvaeri, Marc},
  title   = {Integrability ex machina},
  journal = {Fortschritte der Physik},
  volume  = {69},
  number  = {7},
  pages   = {2100057},
  year    = {2021},
  doi     = {10.1002/prop.202100057}
}

@article{LinChen2025LaxPairFIND,
  author  = {Lin, Shuning and Chen, Yong},
  title   = {{Lax-Pair-FIND}: Discovering {L}ax pair from scarce data via deep learning},
  journal = {Chaos: An Interdisciplinary Journal of Nonlinear Science},
  volume  = {35},
  number  = {11},
  pages   = {113120},
  year    = {2025},
  doi     = {10.1063/5.0278425}
}

@article{PuChen2024LPNN,
  author  = {Pu, Juncai and Chen, Yong},
  title   = {Lax pairs informed neural networks solving integrable systems},
  journal = {Journal of Computational Physics},
  volume  = {510},
  pages   = {113090},
  year    = {2024},
  doi     = {10.1016/j.jcp.2024.113090},
  note    = {Preprint \texttt{arXiv:2401.04982}}
}

@article{deKosterWahls2024,
  author  = {de Koster, Pascal and Wahls, Sander},
  title   = {Data-driven identification of the spectral operator in {AKNS} {L}ax pairs using conserved quantities},
  journal = {Wave Motion},
  volume  = {127},
  pages   = {103273},
  year    = {2024},
  doi     = {10.1016/j.wavemoti.2024.103273}
}

@article{FriedlanderVishik1990,
  author  = {Friedlander, Susan and Vishik, Misha M.},
  title   = {Lax pair formulation for the {E}uler equation},
  journal = {Physics Letters A},
  volume  = {148},
  number  = {6--7},
  pages   = {313--319},
  year    = {1990},
  doi     = {10.1016/0375-9601(90)90809-8}
}

@article{Li2001Euler2D,
  author  = {Li, Yanguang Charles},
  title   = {A {L}ax pair for the two-dimensional {E}uler equation},
  journal = {Journal of Mathematical Physics},
  volume  = {42},
  number  = {8},
  pages   = {3552--3553},
  year    = {2001},
  doi     = {10.1063/1.1378791}
}

@article{Childress3DEuler,
  author  = {Li, Yanguang Charles and Yurov, Artyom V.},
  title   = {Lax pairs and {D}arboux transformations for {E}uler equations},
  journal = {Studies in Applied Mathematics},
  volume  = {111},
  number  = {1},
  pages   = {101--113},
  year    = {2003},
  doi     = {10.1111/1467-9590.t01-1-00229}
}

@article{GerbeauLombardi2014,
  author  = {Gerbeau, Jean-Fr\'ed\'eric and Lombardi, Damiano},
  title   = {Approximated {L}ax pairs for the reduced order integration of nonlinear evolution equations},
  journal = {Journal of Computational Physics},
  volume  = {265},
  pages   = {246--269},
  year    = {2014},
  doi     = {10.1016/j.jcp.2014.01.047}
}

@article{TuritsynPrilepskyEtAl2017,
  author  = {Turitsyn, Sergei K. and Prilepsky, Jaroslaw E. and Le, Son Thai and Wahls, Sander and Frumin, Leonid L. and Kamalian, Morteza and Derevyanko, Stanislav A.},
  title   = {Nonlinear {F}ourier transform for optical data processing and transmission: advances and perspectives},
  journal = {Optica},
  volume  = {4},
  number  = {3},
  pages   = {307--322},
  year    = {2017},
  doi     = {10.1364/OPTICA.4.000307}
}

@article{Kotlyar2020NFTneural,
  author  = {Kotlyar, Oleksandr and Pankratova, Maryna and Kamalian-Kopae, Morteza and Vasylchenkova, Anastasiia and Prilepsky, Jaroslaw E. and Turitsyn, Sergei K.},
  title   = {Combining nonlinear {F}ourier transform and neural network-based processing in optical communications},
  journal = {Optics Letters},
  volume  = {45},
  number  = {13},
  pages   = {3462--3465},
  year    = {2020},
  doi     = {10.1364/OL.394115}
}

@article{Sedov2025NFTNumerics,
  title={Numerical Approaches in Nonlinear Fourier Transform-Based Signal Processing for Telecommunications},
  author={Sedov, Egor and Chekhovskoy, Igor and Fedoruk, Mikhail and Turitsyn, Sergey},
  journal={Studies in Applied Mathematics},
  volume={154},
  number={1},
  pages={e12795},
  year={2025},
  publisher={Wiley Online Library}
}

@article{TeutschBuhrenWaseda2023,
  author  = {Lee, Yu-Chen and Br{\"u}hl, Markus and Doong, Dong-Jiing and Wahls, Sander},
  title   = {Nonlinear {F}ourier classification of 663 rogue waves measured in the {P}hilippine {S}ea},
  journal = {PLOS ONE},
  volume  = {19},
  number  = {5},
  pages   = {e0301709},
  year    = {2024},
  doi     = {10.1371/journal.pone.0301709}
}

@article{Osborne2020NLFA,
  author  = {Osborne, Alfred R.},
  title   = {Nonlinear {F}ourier analysis: rogue waves in numerical modeling and data analysis},
  journal = {Journal of Marine Science and Engineering},
  volume  = {8},
  number  = {12},
  pages   = {1005},
  year    = {2020},
  doi     = {10.3390/jmse8121005}
}

@article{BridgmanHeremanQuispelvanderKamp2013,
  author  = {Bridgman, Terry and Hereman, Willy and Quispel, G. R. W. and van der Kamp, Peter H.},
  title   = {Symbolic computation of {L}ax pairs of partial difference equations using consistency around the cube},
  journal = {Foundations of Computational Mathematics},
  volume  = {13},
  number  = {4},
  pages   = {517--544},
  year    = {2013},
  doi     = {10.1007/s10208-012-9133-9}
}

@article{JinLevermoreMcLaughlin1999,
  author  = {Jin, Shan and Levermore, C. David and McLaughlin, David W.},
  title   = {The semiclassical limit of the defocusing {NLS} hierarchy},
  journal = {Communications on Pure and Applied Mathematics},
  volume  = {52},
  number  = {5},
  pages   = {613--654},
  year    = {1999},
  doi     = {10.1002/(SICI)1097-0312(199905)52:5<613::AID-CPA2>3.0.CO;2-L}
}

@article{Krichever1977FA,
  author  = {Krichever, I. M.},
  title   = {Integration of nonlinear equations by the methods of algebraic geometry},
  journal = {Functional Analysis and Its Applications},
  volume  = {11},
  number  = {1},
  pages   = {12--26},
  year    = {1977},
  doi     = {10.1007/BF01135528}
}

@article{BobenkoReymanSemenov1989,
  author  = {Bobenko, A. I. and Reyman, A. G. and Semenov-Tian-Shansky, M. A.},
  title   = {The {Kowalewski} top 99 years later: {A} {L}ax pair, generalizations and explicit solutions},
  journal = {Communications in Mathematical Physics},
  volume  = {122},
  number  = {2},
  pages   = {321--354},
  year    = {1989},
  doi     = {10.1007/BF01257419}
}

@article{GoktasHereman1997,
  author  = {G\"okta\c{s}, \"Unal and Hereman, Willy},
  title   = {Symbolic computation of conserved densities for systems of nonlinear evolution equations},
  journal = {Journal of Symbolic Computation},
  volume  = {24},
  number  = {5--6},
  pages   = {591--621},
  year    = {1997},
  doi     = {10.1006/jsco.1997.0154}
}

@book{BabelonBernardTalon2003,
  author    = {Babelon, Olivier and Bernard, Denis and Talon, Michel},
  title     = {Introduction to Classical Integrable Systems},
  series    = {Cambridge Monographs on Mathematical Physics},
  publisher = {Cambridge University Press},
  address   = {Cambridge},
  year      = {2003},
  doi       = {10.1017/CBO9780511535024}
}

@article{AKNS1974,
  author  = {Ablowitz, M. J. and Kaup, D. J. and Newell, A. C. and Segur, H.},
  title   = {The inverse scattering transform--{F}ourier analysis for nonlinear problems},
  journal = {Studies in Applied Mathematics},
  volume  = {53},
  number  = {4},
  pages   = {249--315},
  year    = {1974},
  doi     = {10.1002/sapm1974534249}
}

@article{Doikou2012,
  author  = {Doikou, Anastasia},
  title   = {Selected Topics in Classical Integrability},
  journal = {International Journal of Modern Physics A},
  volume  = {27},
  number  = {5},
  pages   = {1230003},
  year    = {2012},
  doi     = {10.1142/S0217751X12300037},
}

@book{DrazinJohnson1989,
  author    = {Drazin, P. G. and Johnson, R. S.},
  title     = {Solitons: An Introduction},
  series    = {Cambridge Texts in Applied Mathematics},
  publisher = {Cambridge University Press},
  address   = {Cambridge},
  year      = {1989},
  doi       = {10.1017/CBO9781139172059}
}

@misc{LeeWahls2025,
  author        = {Lee, Yu-Chen and Wahls, Sander},
  title         = {Field observation of soliton gases in the deep open ocean},
  note          = {arXiv preprint},
  eprint        = {2510.04662},
  archivePrefix = {arXiv},
  year          = {2025}
}

@misc{KonopelchenkoPC,
  author = {Konopelchenko, B.},
  note = {private communication},
  year = {2026}
}

\appendix
\addtocontents{section}{Appendix}
\section*{Appendix}
\def\thesection{A}

\begin{algorithm}[b!]
\caption{Conservation hierarchy for $u_t = q_x$ via the first-order Lax pair $(u + \partial_x,\, q)$.}
\label{alg:cq_generator}
\begin{algorithmic}[1]
\Require Conservation-form PDE $u_t = q_x$. Target subscript-sum level $N \geq 1$.
\Ensure A basis $\{F_\beta\}$ for the space of conservation densities at level $\leq N$.
\State Set $L \gets u + \partial_x$, \, $P \gets q$, \, $\rho_0 \gets 1$.
\For{$s = 1, \ldots, N$} \Comment{Bell polynomials}
  \State $\rho_s \gets u\,\rho_{s-1} + (\rho_{s-1})_x$.
\EndFor
\For{$j = 0, \ldots, N$} \Comment{Spatial derivatives of $P$}
  \State $P^{(j)} \gets D_x^j P$.
\EndFor
\For{$s = 1, \ldots, N$} \Comment{Time derivatives via \eqref{eq:dt_rho}}
  \State $\partial_t\rho_s \gets \sum_{j=1}^{s} \binom{s}{j}\, \rho_{s-j}\, P^{(j)}$.
\EndFor
\State Enumerate partitions $\alpha$ with $|\alpha| \leq N$; form $\rho$-monomials $m_\alpha = \rho_{i_1}\cdots\rho_{i_k}$.
\State Reduce mod $D_x$ via \eqref{eq:Drho}; pick representatives $\{m_{\alpha_1}, \ldots, m_{\alpha_n}\}$.
\For{$\alpha \in \{\alpha_1, \ldots, \alpha_n\}$}
  \State Compute $\partial_t m_\alpha$ by Leibniz from $\partial_t \rho_s$.
  \State Reduce mod $D_x$ via integration by parts to $\partial_t m_\alpha \equiv \sum_{k=1}^{m} M_{k\alpha} R_k$.
\EndFor
\State Form the $m \times n$ matrix $M = (M_{k\alpha})$ over $\Q$.
\State Compute $\ker M = \{c \in \Q^n : Mc = 0\}$.
\For{each basis vector $c^{(\beta)} \in \ker M$}
  \State $F_\beta \gets \sum_{\alpha} c_\alpha^{(\beta)}\, m_\alpha$.
\EndFor
\State \Return $\{F_\beta\}$. Each $\int F_\beta\, dx$ is conserved.
\end{algorithmic}
\end{algorithm}

\subsection{KdV CQ Algorithm}\label{app:KdVCQ}
Algorithm~\ref{alg:cq_generator} summarizes the approach of section~\ref{sec:kdv} to the conservation laws of the KdV equation. For further clarity and sake of concreteness, we run the above algorithm on the KdV equation, with $P = 3u^2 + u_{xx}$. We use the equivalence symbol $\equiv$ to denote equality modulo total $D_x$ derivatives, and we work either on a periodic domain or on $\R$ with rapid decay.

\paragraph{Level 1.} The single class $[\rho_1]$ has $\partial_t \rho_1 = P^{(1)} = (3u^2 + u_{xx})_x$, a total $D_x$. Conserved direction:
\[
I_1 = \int u\, dx.
\]

\paragraph{Level 2.} The new class is $[\rho_2]$, with $[\rho_1^2] = [\rho_2]$ via \eqref{eq:Drho}. Compute
\[
\partial_t \rho_2 = 2\rho_1 P^{(1)} + P^{(2)}.
\]
Both terms reduce: $P^{(2)} = D_x P^{(1)}$ is a total derivative, and $2\rho_1 P^{(1)} = (2uP)_x - 2 u_x P$ with $-2 u_x P = -2 u_x(3u^2 + u_{xx}) = -2(u^3)_x - (u_x^2)_x$. New conserved direction:
\[
I_2 = \int u^2\, dx.
\]

\paragraph{Level 3.} The new class $[\rho_3]$ is the unique level-3 representative (since $[\rho_3] = [\rho_1\rho_2] = [\rho_1^3]$). Direct computation of
\[
\partial_t \rho_3 = 3\rho_2 P^{(1)} + 3\rho_1 P^{(2)} + P^{(3)}
\]
and reduction by integration by parts give $\partial_t \rho_3 \equiv 3 u_x^3$, which is not a total derivative. Level 3 contributes no new direction.

\paragraph{Level 4.} New representatives mod $D_x$: $[\rho_1\rho_3]$, $[\rho_2^2]$, $[\rho_1^4]$, with $[\rho_4] = [\rho_1\rho_3]$ and $[\rho_1^2\rho_2] = [\rho_1^4]$ from \eqref{eq:Drho}. Combined with the level-3 carry-over $[\rho_3]$, the candidate space is four-dimensional. The reduced time derivatives sit in a two-dimensional residue space spanned by $u_x^3$ and $u u_x^3$:
\[
\partial_t \rho_3 \equiv 3 u_x^3, \quad \partial_t (\rho_1\rho_3) \equiv 12 u u_x^3 - 6 u_x^3, \quad \partial_t \rho_1^4 \equiv 12 u u_x^3, \quad \partial_t \rho_2^2 \equiv 12 u u_x^3 + 6 u_x^3.
\]
Only one linear combination kills both residues. The unique new kernel direction is
\[
F_3 = \rho_1^3 + \rho_1^2 \rho_2 - \tfrac{1}{2}\rho_2^2 - \tfrac{1}{2}\rho_1^4 \;\equiv\; u^3 - \tfrac{1}{2}u_x^2,
\]
giving
\be
\boxed{\;I_3 = \int\bigl(u^3 - \tfrac{1}{2}u_x^2\bigr)\, dx.\;}
\ee

\paragraph{Level 5.} The output residues from new level-5 candidates are independent of every residue seen at lower levels. The kernel does not extend.

\paragraph{Level 6.} The kernel computation, with new candidates at subscript sum~$6$ and carry-overs from level~$4$, returns one new direction:
\be
\boxed{\;I_4 = \int\bigl(u^4 - 2 u u_x^2 + \tfrac{1}{5} u_{xx}^2\bigr)\, dx,\;}
\ee
or equivalently, after clearing denominators, $5 I_4 = \int\bigl(5 u^4 - 10 u u_x^2 + u_{xx}^2\bigr)\, dx$. The translation to $\rho$-form via \eqref{eq:rho_gauge} and \eqref{eq:inverse_bell} is mechanical: $u^4 = \rho_1^4$, $u u_x^2 = \rho_1\rho_2^2 - 2\rho_1^3\rho_2 + \rho_1^5$, and $u_{xx}^2 = \rho_3^2 - 6\rho_1\rho_2\rho_3 + 4\rho_1^3\rho_3 + 9\rho_1^2\rho_2^2 - 12\rho_1^4\rho_2 + 4\rho_1^6$. Substituting and collecting yields a finite combination of integrated $\rho$-monomials at subscript sum $\leq 6$.

\paragraph{Level 7.} Same story as level 5. No new direction.

\paragraph{Level 8.} The first new direction beyond $I_1, \ldots, I_4$ appears here. The kernel computation, on the ansatz $u^5 + a_1 u^2 u_x^2 + a_2 u u_{xx}^2 + a_3 u_{xxx}^2$ with the requirement that its time derivative reduces to zero, returns the unique solution $(a_1, a_2, a_3) = (-5,\, 1,\, -1/14)$. After clearing denominators,
\be
\boxed{\;I_5 = \int\bigl(14 u^5 - 70 u^2 u_x^2 + 14 u u_{xx}^2 - u_{xxx}^2\bigr)\, dx.\;}
\ee

\paragraph{The pattern.}
At each subscript-sum level, a finite linear system over $\Q$ either extends the kernel by one direction or returns nothing. The new directions, in order, are $I_1, I_2, I_3, I_4, I_5$, appearing at levels $1, 2, 4, 6, 8$. The cross-product structure (the coefficient $-\tfrac12$ in $I_3$, the coefficients $-2$ and $\tfrac15$ in $I_4$) is the unique direction in the level-$N$ candidate space whose time derivative vanishes modulo $D_x$, identified by row-reducing a matrix of $\Q$-coefficients.

\subsection{The Butler--Hay Test for the First-Order KdV Pair}
\label{app:butlerhay}

The Butler--Hay test~\cite{ButlerHay} gives two diagnostics for fake Lax pairs.
The \emph{gauge test} seeks a transformation $\varphi \mapsto S\varphi$ conjugating
$L$ to constant coefficients in $u$; if one exists, the pair is fake, since the
gauged spatial problem could belong to any PDE on the same substrate. The
\emph{excess-freedom} ($u$-fake) test replaces each $u$-dependent coefficient by an
arbitrary function and imposes compatibility; leftover freedom signals a fake,
whereas a genuine substrate fixes $P$ rigidly. We run both on
\eqref{eq:fo_pair_clean}.

The gauge test is immediate. The presentation $L = e^{-U}\partial_x e^U$ of
\eqref{eq:gauge_L}, read backwards with $S = e^U$, gives
\[
S L S^{-1} = e^U (u + \partial_x)\,e^{-U} = \partial_x,
\]
constant-coefficient and free of $u$. The gauged eigenvalue problem
$\partial_x\tilde\varphi = \lambda\tilde\varphi$ has solutions
$\tilde\varphi = e^{\lambda x}$ carrying no information about $u$, and would govern
any PDE on the substrate $(u+\partial_x,\,q)$. The pair is fake.

The excess-freedom test agrees. Since $L = u + \partial_x$ gives $[L,P] = P_x$ for
any multiplication operator $P$, compatibility reads $u_t = P_x$ and holds for
every conservation-form PDE $u_t = q_x$ with $P = q$. The flux is unconstrained,
and KdV with $q = 3u^2 + u_{xx}$ is one member of an infinite family sharing the
substrate. This is the conservation-law dressing of Calogero and
Nucci~\cite{CalogeroNucci}, whose fake pair carries arbitrary $f, g$ with
compatibility $f_t = g_x$; the first-order pair reported in~\cite{SILO} is one
realization, and Sakovich~\cite{Sakovich2001} recast the same fact as a
gauge invariant via the cyclic basis.

By contrast, the classical pair $L_{\mathrm{Schr}} = -u - \partial_x^2$,
$P_{\mathrm{Schr}} = -4\partial_x^3 - 6u\partial_x - 3u_x$ passes both tests
substantively: conjugating a second-order operator by $S = e^{\phi}$ produces a
first-order term $-2\phi_x\partial_x$, so removing it forces $\phi_x = 0$ and
leaves $u$ intact, and $\partial_t L_{\mathrm{Schr}} = [L_{\mathrm{Schr}},
P_{\mathrm{Schr}}]$ fixes $P_{\mathrm{Schr}}$ up to gauge with no alternative flux.
The first-order pair is fake by every gauge-invariant measure here, in agreement
with the Krichever collapse of Appendix~\ref{app:krichever} and the spectral
findings of Section~\ref{sec:spectral}, yet Section~\ref{sec:operator_algebra}
recovers the full hierarchy from it. 

\subsection{The Krichever Test for the First-Order KdV Pair}
\label{app:krichever}

The Krichever test~\cite{Krichever1977, Krichever1977FA} is a necessary spectral
condition for algebro-geometric content. Linearizing $u_t = K[u]$ and its auxiliary
system $L\varphi = \lambda\varphi$, $\varphi_t = -P\varphi$ at a constant background
$u = u_0$ yields two dispersion laws, $\omega_{\mathrm{PDE}}(\xi)$ from the plane
wave $\delta u = e^{i(\xi x - \omega t)}$ and $\omega_{\mathrm{Lax}}(\xi)$ from the
temporal frequency of eigenfunctions of $L$; a genuine pair requires
$\omega_{\mathrm{Lax}}$ polynomial on the spectral curve and equal to
$\omega_{\mathrm{PDE}}$ up to normalization.

For KdV, $\delta u_t = 6u_0\,\delta u_x + \delta u_{xxx}$ gives
\begin{equation}\label{eq:kdv_dispersion}
\omega_{\mathrm{KdV}}(\xi) = \xi^3 - 6 u_0\,\xi.
\end{equation}
For the first-order pair, $L\varphi = \mu\varphi$ at $u_0$ gives
$\varphi = e^{(\mu - u_0)x}$, so $\mu = u_0 + i\xi$ and $\varphi = e^{i\xi x}$ on the
linear curve $\mu - u_0 = i\xi$. There $u_{xx} = 0$ and $P = 3u^2 + u_{xx}$
collapses to the constant $3u_0^2$, so $\varphi_t = -3u_0^2\varphi$ gives
\begin{equation}\label{eq:fo_lax_dispersion}
\omega_{\mathrm{Lax}}(\xi) = -3 i u_0^2,
\end{equation}
constant in $\xi$ (the imaginary value reflects the non-self-adjointness of $L$).
Against the cubic \eqref{eq:kdv_dispersion} this is polynomial only degenerately, as
a degree-zero constant, and recovers none of the dispersive $\xi^3$ structure of KdV
for any $u_0$.

The classical pair passes substantively. At $u_0$ the eigenfunctions of
$L_{\mathrm{Schr}} = -\partial_x^2 - u$ are $\varphi = e^{i\xi x}$ with
$\lambda = \xi^2 - u_0$, and $P_{\mathrm{Schr}}\varphi = i(4\xi^3 - 6u_0\xi)\varphi$
gives $\omega_{\mathrm{Lax}}^{\mathrm{Schr}}(\xi) = 4\xi^3 - 6u_0\xi$, cubic on the
curve and matching \eqref{eq:kdv_dispersion} up to the cubic normalization. The
first-order pair carries none of this dispersive content.
\end{document}